\documentclass[envcountsect,envcountsame,oribibl,orivec]{llncs}
\usepackage{amssymb}
\usepackage{amsmath}
\usepackage{url}
\usepackage[right=3cm,left=3cm]{geometry}

\usepackage{enumitem}
\usepackage{colonequals}
\usepackage{MnSymbol}

\usepackage{graphics}
\usepackage{graphicx}
\usepackage{hyperref}

\usepackage{etoolbox}
\makeatletter
\let\llncs@addcontentsline\addcontentsline
\patchcmd{\maketitle}{\addcontentsline}{\llncs@addcontentsline}{}{}
\patchcmd{\maketitle}{\addcontentsline}{\llncs@addcontentsline}{}{}
\patchcmd{\maketitle}{\addcontentsline}{\llncs@addcontentsline}{}{}
\makeatother
\usepackage{bookmark}
\makeatletter
\newcommand{\todo}[1]{\marginpar{\textbf{TODO\footnotemark}}\@latex@warning{TODO: #1}\footnotetext{ #1}}
\makeatother
\usepackage{tikz}
\newcommand*\circled[1]{\tikz[baseline=(char.base)]{
		\node[shape=circle,draw,inner sep=0.4pt] (char) {#1};}}

\newcommand{\Prop}{\textsf{Prop}}
\newcommand{\Formulae}{\textsf{Fm}}

\newcommand{\lfalse}{\bot}
\newcommand{\lneg}{\neg}

\newcommand{\propax}{\ensuremath{(\textsf{Taut})}}
\newcommand{\lrule}[2]{\displaystyle{\frac{#1}{#2}}}
\newcommand{\mprule}{\ensuremath{(\textsf{MP})}}
\newcommand{\limplies}{\rightarrow}
\newcommand{\liff}{\leftrightarrow}

\newcommand{\lset}[1]{\{ #1 \}}

\newcommand{\lnext}{\bigcirc}
\newcommand{\lalways}{\Box}
\newcommand{\leventually}{\Diamond}
\newcommand{\lonce}{\diamondminus}
\newcommand{\lsofar}{\boxminus}

\newcommand{\luntil}{{\,\mathcal{U}\,}}
\newcommand{\lsince}{{\,\mathcal{S}\,}}
\newcommand{\ltime}{{\,\texttt{time}\,}}
\newcommand{\ltrue}{{\,\texttt{true}\,}}
\newcommand{\lalwaysPastFuture}{\boxdot}

\newcommand{\lobligatory}{{\,\mathcal{O}}}
\newcommand{\lpermissible}{{\,\mathcal{P}}}
\newcommand{\lepistemic}{{\,\mathcal{E}}}

\newcommand{\lwprevious}{\circled{\textit{w}}}
\newcommand{\lsprevious}{\circled{\textit{s}}}

\newcommand{\lunless}{{\,\mathcal{W}\,}}

\newcommand{\kax}{\ensuremath{\textsf{-k}}}
\newcommand{\nextkax}{\ensuremath{(\lnext\kax)}}
\newcommand{\alwayskax}{\ensuremath{(\lalways\kax)}}
\newcommand{\funax}{\ensuremath{(\textsf{fun})}}

\newcommand{\indax}{\ensuremath{(\textsf{ind})}}
\newcommand{\uoneax}{\ensuremath{(\luntil\textsf{1})}}
\newcommand{\utwoax}{\ensuremath{(\luntil\textsf{2})}}
\newcommand{\necrule}{\ensuremath{\textsf{-nec}}}
\newcommand{\nextnecrule}{\ensuremath{(\lnext\necrule)}}
\newcommand{\alwaysnecrule}{\ensuremath{(\lalways\necrule)}}
\newcommand{\prevnecrule}{\ensuremath{(\lwprevious\necrule)}}
\newcommand{\sofarnecrule}{\ensuremath{(\lsofar\necrule)}}

\newcommand{\alwaysPastFuturenecrule}{\ensuremath{(\lalwaysPastFuture\necrule)}}
\newcommand{\prevRMrule}{\ensuremath{(\lwprevious\textsf{-RM})}}
\newcommand{\onceRMrule}{\ensuremath{(\lonce\textsf{-RM})}}
\newcommand{\prevkax}{\ensuremath{(\lwprevious\kax)}}
\newcommand{\sofarkax}{\ensuremath{(\lsofar\kax)}}
\newcommand{\swprevax}{\ensuremath{(\textsf{sw})}}
\newcommand{\fpax}{\ensuremath{(\textsf{FP})}}
\newcommand{\pfax}{\ensuremath{(\textsf{PF})}}
\newcommand{\initialax}{\ensuremath{(\textsf{initial})}}
\newcommand{\sofarindax}{\ensuremath{(\lsofar\textsf{-ind})}}
\newcommand{\soneax}{\ensuremath{(\lsince\textsf{1})}}
\newcommand{\stwoax}{\ensuremath{(\lsince\textsf{2})}}

\newcommand{\JTO}{\textsf{JTO}}

\newcommand{\LPLTLp}{\textsf{LPLTL}^{\sf P}}

\newcommand{\Logic}[1]{\mathsf{#1}} 

\newcommand{\CTerms}{\textsf{Const}}
\newcommand{\VTerms}{\textsf{Var}}
\newcommand{\Terms}{\textsf{Tm}}
\newcommand{\Ag}{\textsf{Ag}}

\newcommand{\jbox}[1]{\left[#1\right]\!}
\newcommand{\jboxAgent}[1]{\left[#1\right]_\agent\!}

\newcommand{\jdiamondAgent}[1]{\langle#1\rangle_\agent\,}

\newcommand{\jboxOAgent}[1]{\left[#1\right]^\lobligatory_\agent\!}
\newcommand{\jboxO}[2]{\left[#1\right]^\lobligatory_#2\!}

\newcommand{\jboxPAgent}[1]{\langle#1\rangle^\lpermissible_\agent\!}
\newcommand{\jboxP}[2]{\langle#1\rangle^\lpermissible_#2\!}

\newcommand{\truthsetModel}[1]{\left\|#1\right\|_\system\!}
\newcommand{\truthset}[2]{\left\|#1\right\|_#2\!}
\newcommand{\tapp}{\cdot}
\newcommand{\ttsum}{+}
\newcommand{\tinspect}{!}

\newcommand{\tref}{\ddagger}

\newcommand{\LP}{\textsf{LP}}

\newcommand{\appax}{\ensuremath{(\textsf{application})}}
\newcommand{\sumax}{\ensuremath{(\textsf{sum})}}
\newcommand{\posintax}{\ensuremath{(\textsf{positive introspection})}}

\newcommand{\refax}{\ensuremath{(\textsf{factivity})}}

\newcommand{\conax}{\ensuremath{(\textsf{consistency})}}
\newcommand{\shiftrefax}{\ensuremath{(\textsf{obligated factivity})}}
\newcommand{\nocax}{\ensuremath{(\textsf{no conflicts})}}
\newcommand{\nocstrongax}{\ensuremath{(\textsf{strong no conflicts})}}
\newcommand{\appOax}{\ensuremath{(\textsf{application-}\lobligatory)}}

\newcommand{\iteratedconstnecrule}{\ensuremath{(\textsf{iax}\necrule)}}
\newcommand{\iteratedconstnecruleO}{\ensuremath{(\textsf{iax}\necrule\textsf{-}\lobligatory)}}

\newcommand{\CS}{\textsf{CS}}

\newcommand{\accrel}{R}
\newcommand{\accrelO}{R^\lobligatory}

\newcommand{\numberofagents}{h}
\newcommand{\agent}{i}

\newcommand{\csNeighborhood}{\ensuremath{(\textsf{constant-specification-}\neighborhood)}}
\newcommand{\appNeighborhood}{\ensuremath{(\textsf{application-}\neighborhood)}}
\newcommand{\sumNeighborhood}{\ensuremath{(\textsf{sum-}\neighborhood)}}
\newcommand{\posintNeighborhood}{\ensuremath{(\textsf{positive-introspection-}\neighborhood)}}

\newcommand{\refNeighborhood}{\ensuremath{(\textsf{reflexivity-}\neighborhood)}}

\newcommand{\csNeighborhoodO}{\ensuremath{(\textsf{constant-specification-}\neighborhood^\lobligatory)}}
\newcommand{\appNeighborhoodO}{\ensuremath{(\textsf{application-}\neighborhood^\lobligatory)}}
\newcommand{\nocNeighborhoodO}{\ensuremath{(\textsf{noc-}\neighborhood^\lobligatory)}}
\newcommand{\shiftrefNeighborhoodO}{\ensuremath{(\textsf{obligated-factivity-}\neighborhood^\lobligatory)}}

\newcommand{\runs}{\mathcal{R}}
\newcommand{\system}{\mathcal{I}}
\newcommand{\evidence}{\mathcal{E}}
\newcommand{\evidenceO}{\mathcal{E}^\mathcal{O}}

\newcommand{\valuation}{\nu}
\newcommand{\entails}{\vDash}

\newcommand{\neighborhood}{N_\agent}
\newcommand{\neighborhoodO}{N^\lobligatory_\agent}

\newcommand{\Neigh}{neighborhood }

\newcommand{\N}{\mathbb{N}}
\newcounter{enumsave}
\renewcommand{\phi}{\varphi}

\newcommand{\Subf}{\mathsf{Subf}}
\newcommand{\MCS}{\mathsf{MCS}}

\newcommand{\RO}[4]{#1 R_\lnext #2\ [#3, #4]}

\newcommand{\protagoras}{\texttt{p}}
\newcommand{\euathlus}{\texttt{e}}
\newcommand{\judge}{\texttt{j}}

\newcommand{\agreement}{\texttt{a}}
\newcommand{\verdictP}{\texttt{verdict}_\protagoras}
\newcommand{\verdictE}{\texttt{verdict}_\euathlus}

\newcommand{\winP}{\texttt{win}_\protagoras}
\newcommand{\winE}{\texttt{win-first}_\euathlus}
\newcommand{\pay}{\texttt{pay}}
\newcommand{\sueP}{\texttt{sue}_\protagoras}

\newcommand{\winPsecond}{\texttt{win-second}_\protagoras}
\newcommand{\contract}{\texttt{contract}}
\newcommand{\court}{\texttt{court}}
\newcommand{\NowinE}{\texttt{No-win-first}_\euathlus}
\newcommand{\PsueE}{\texttt{PsueE}}
\newcommand{\PastLooking}{\texttt{Past-looking}}

\begin{document}

\title{A temporal logic of epistemic and normative justifications,\\ with an application to the Protagoras paradox\thanks{A preliminary version of this paper was presented in the ``Annual Seminar on Mathematical Logic and its Applications", Arak University of Technology, September 2019. The author would like to thank the organizers for their invitation.}
}
\author{Meghdad Ghari\thanks{This research was in part supported by a grant from IPM (No. 98030426)  and carried out in IPM-Isfahan Branch.}
}
\institute{Department of Philosophy, Faculty of Literature and Humanities,\\
	University of Isfahan, Isfahan, Iran \\ and \\ School of Mathematics,
	Institute for Research in Fundamental Sciences (IPM), \\ P.O.Box: 19395-5746, Tehran, Iran \\ \email{ghari@ipm.ir}
}

\maketitle

\begin{abstract}
We combine linear temporal logic (with both past and future modalities) with a deontic version of justification logic to provide a framework for reasoning about time and epistemic and normative reasons. In addition to temporal modalities, the resulting logic contains two kinds of justification assertions: epistemic justification assertions and deontic justification assertions. The former presents justification for the agent's knowledge and the latter gives reasons for why a proposition is obligatory. We present two kinds of semantics for the logic: one based on Fitting models and the other based on neighborhood models. The use of neighborhood semantics enables us to define the dual of deontic justification assertions properly, which corresponds to the notion of permission in deontic logic. We then establish the soundness and completeness of an axiom system  of the logic with respect to these semantics. Further, we formalize the Protagoras versus Euathlus paradox in this logic and present a precise analysis of the paradox, and also briefly discuss Leibniz's solution.
\end{abstract}


\section{Introduction}

Justification logics are modal-like logics that provide a framework for reasoning about epistemic justifications (see \cite{Art-Fit-Book-2019,Kuz-Stu-Book-2019} and also \cite{Art08RSL,ArtFit11SEP} for a survey). The language of multi-agent justification logics extends the language of propositional logic by justification terms and expressions of the form $\jbox{t}_\agent \phi$, with the intended meaning ``$t$   is agent $i$'s justification for knowing (or believing) $\phi$.''  The \emph{Logic of Proofs} $\LP$ was the first logic in the family of justification logics, introduced by Artemov in \cite{Art95TR,Art01BSL}. The logic of proofs is a justification counterpart of the modal epistemic logic $\Logic{S4}$. Various extensions and variants of the logic of proofs have been introduced so far. For example, a deontic reading of justification assertions are recently mentioned in the literature, in which $\jbox{t}_\agent \phi$ is read ``$t$ is a reason why $\phi$ is obligatory for agent $i$''. In this respect, various deontic justification logics have been introduced (cf. \cite{Faroldi-Protopopescu-IGPL-2019,Carneiro-2019,Ghari-DEON-2021}). On the other hand, extensions of justification logics with temporal modalities are introduced in \cite{Bucheli15,BucheliGhariStuder2017,Gha18temporal-arXiv,Ghari-IGPL-2021}. Temporal justification logics are a new family of temporal logics of knowledge in which the knowledge of agents is modeled using a justification logic.

The importance of combining deontic logic with temporal logic is argued in many papers (see \cite{Thomason1981}, from an extensive literature). Some authors argued that attention to time is crucial in handling some deontic puzzles, such as the Chisholm puzzle. Another useful combination is obtained by adding a deontic logic to a logic of knowledge (see e.g. \cite{Lomuscio-Sergot-SL-2003,Pacuit-Parikh-Cogan-Synthese-2006}). 

This paper aims at combining a temporal justification logic with a deontic justification logic. We call the resulting system the logic of justification, time, and obligation, denoted by $\JTO$. The temporal justification logic part of $\JTO$, that comes from \cite{Ghari-IGPL-2021}, is the logic $\LPLTLp$ which is a combination of linear temporal logic (with both past and future modalities) with the logic of proofs. The deontic justification logic part of $\JTO$ is an extension of a logic introduced in \cite{Ghari-DEON-2021}.

In addition to temporal modalities, the logic $\JTO$ contains two kinds of justification assertions: epistemic justification assertions $\jboxAgent{t} \phi$ and deontic (or normative) justification assertions $\jboxOAgent{t} \phi$, which are read respectively as  ``$t$ is agent $i$'s justification for $\phi$'' and ``$t$ is a reason why $\phi$ is obligatory for agent $i$'' (or ``it is obligatory for agent $i$ that $\phi$ is true because of $t$''). Epistemic justification assertions originated in the known multi-agent justification logics \cite{TYav08TOCS,BucKuzStu11JANCL,Ghari2014}, and deontic justification assertions originated in deontic justification logics \cite{Faroldi-Protopopescu-IGPL-2019,Carneiro-2019,Ghari-DEON-2021}. 
Let us recall the notion of normative reason: a normative reason for an agent to do an action is an explanation of why the agent ought to do that action (cf. \cite{Faroldi-Protopopescu-IGPL-2019}).
We present two kinds of semantics for the logic $\JTO$: one based on Fitting models (similar to that given in \cite{BucheliGhariStuder2017,Ghari-IGPL-2021}) and the other based on neighborhood models. Both semantics are extensions of interpreted systems that are widely used to model knowledge and time in multi-agent systems (see, for instance  \cite{FHMV95}). The use of  neighborhood semantics enables us to define the dual of justification assertions, i.e. $\jdiamondAgent{t} \phi$ and $\jboxPAgent{t} \phi$, which are read  respectively as ``$t$ is a reason why $\phi$ is compatible with agent $i$'s knowledge'' and ``$t$ is a reason why $\phi$ is permitted for agent $i$''. The latter corresponds to the notion of permission which is of great importance in deontic logic.

In order to show how the logic $\JTO$ can be used in practice, we consider a concrete case study and formalize the Protagoras versus Euathlus case paradox. This paradox is formalized in  \cite{Lenzen1977,Aqvist1995,Glavanicova-Pascucci-DEON-2021} (Lenzen and Aqvist formalization is restated in \cite{Lukowski2011}). In order to to formalize the paradox, Lenzen and Aqvist use a modal logic ${\sf S5}$ (with axioms of identity predicate), in which neither deontic modalities nor temporal operators are used. Glavanicova and Pascucci, in contrast, formalize the paradox in an interval based temporal logic which contains deontic modal operators. All aforementioned papers use the necessity modality to formalize the contract between Protagoras and Euathlus. Glavanicova and Pascucci focus on different normative sources that is used in the arguments of Protagoras and Euathlus. Instead of normative sources, the language of $\JTO$ can express different normative reasons that are used in the arguments. We then show that this vanishes the appearance of contradiction in a formalization of the paradox in $\JTO$. 

A section-by-section outline of the content of this paper follows. Section \ref{sec:Syntax} introduces the language of our logic $\JTO$, and Section \ref{sec:Axioms}  presents axiom systems for the logic. Section \ref{sec:Semantics Fiittng} introduces a semantics based on interpreted systems and Fitting models for $\JTO$.  Then, Section \ref{sec:Completeness} gives details of the completeness proof for $\JTO$. Section \ref{sec:Semantics Fiittng} introduces a semantics based on interpreted systems and neighborhood models for $\JTO$ and provides also the completeness theorem with respect to these models. Finally Section \ref{sec: Protagoras Paradox} discusses the Protagoras Paradox and its formalization in the logic $\JTO$.


\section{Language}
\label{sec:Syntax}

In the following, let $\Ag$ be a finite set of agents, $\CTerms$ be a countable set of justification constants, $\VTerms$ be a countable set of justification variables, and $\Prop$ be a countable set of atomic propositions.

Justification terms and formulas are constructed by the following mutual grammar:
\begin{gather*}
t \coloncolonequals c \mid x \mid \; \tinspect t \mid \; t \ttsum t \mid \; t \tapp t \, , \\
s \coloncolonequals c \mid x \mid \;  \tref s \mid \; s \tapp s \, , \\
\phi \coloncolonequals p \mid \lfalse \mid \phi \limplies \phi \mid \lnext \phi \mid \lwprevious \phi \mid \phi \luntil \phi \mid \phi \lsince \phi \mid \jboxAgent{t} \phi \mid \jboxOAgent{s} \phi \, .
\end{gather*}
where $c \in \CTerms$, $x \in \VTerms$, $i \in \Ag$, and $p \in \Prop$.  The set of justification terms and formulas are denoted by  $\Terms$ and $\Formulae$ respectively. The set of all terms that are used in epistemic justification assertions, defined by the first grammar above, are denoted by $\Terms^\lepistemic$, and the set of all terms that are used in deontic justification assertions, defined by the second grammar above, are denoted by $\Terms^\lobligatory$. Thus, $\Terms = \Terms^\lepistemic \cup \Terms^\lobligatory$.

Note that the only common part of $\Terms^\lepistemic$ and $\Terms^\lobligatory$ are justification constants and justification variables. Therefore, although for $x \in \VTerms$ both $\jboxAgent{x} p$ and $\jboxOAgent{x } p$ are well-formed formulas, the following expressions are not well-formed formulas: $\jboxAgent{\tref x} p$, $\jboxOAgent{x \ttsum y} p$, $\jboxOAgent{\tinspect x} p$, where $x,y \in \VTerms$. From now on when we write $\jboxAgent{t} \phi$ or $\jboxOAgent{s} \phi$ we will assume that $t \in \Terms^\lepistemic$ and $s \in \Terms^\lobligatory$.

The temporal operators $\lnext,\lwprevious, \luntil,\lsince$ are respectively called \textit{next (or tomorrow), weak previous (or weak yesterday), until}, and \textit{since}. We use the following usual abbreviations:
\begin{align*}
\lneg \phi &\colonequals \phi \limplies \lfalse &
\top &\colonequals \lneg \lfalse &\\
\phi \lor \psi &\colonequals \lneg \phi \limplies \psi &
\phi \land \psi &\colonequals \lneg (\lneg \phi \lor \lneg \psi) \\
\phi \liff \psi &\colonequals (\phi \limplies \psi) \land (\psi \limplies \phi) & \lsprevious \phi &\colonequals \neg \lwprevious \neg \phi
\\
\leventually \phi &\colonequals \top \luntil \phi & \lalways \phi &\colonequals \lneg \leventually \lneg \phi  \\
\lonce \phi &\colonequals \top \lsince \phi  & \lsofar \phi &\colonequals \lneg \lonce\lneg \phi\\
\lalwaysPastFuture \phi &\colonequals \lsofar \phi \wedge \lalways \phi & \phi \lunless \psi &\colonequals  (\phi \luntil \psi) \vee (\lalways \phi) \\
\jdiamondAgent{t} \phi &\colonequals \neg \jboxAgent{t} \neg \phi  & \jboxPAgent{t} \phi &\colonequals \neg \jboxOAgent{t} \neg \phi.
\end{align*}
The temporal operators $\lsprevious, \lalways, \leventually, \lsofar, \lonce, \lalwaysPastFuture, \lunless$ are respectively called \textit{strong previous, always from now on (or henceforth), sometime (or eventuality), has-always-been},  \textit{once},  \textit{always}, and \textit{unless} (or \textit{weak until}). The connectives $\lwprevious, \lsprevious, \lsofar, \lonce,\lsince$ are called the past modalities, while the connectives $\lnext, \lalways, \leventually, \luntil, \lunless$ are called the future modalities. Associativity and precedence of connectives, as well as the corresponding omission of brackets, are handled in the usual manner.
%
%

As stated in Introduction, the formulas $\jboxAgent{t} \phi$ and  $\jboxOAgent{t} \phi$ say respectively that  ``$t$ is agent $i$'s justification for $\phi$'' and ``$t$ is a reason why $\phi$ is obligatory for agent $i$'', and their dual $\jdiamondAgent{t} \phi$ and $\jboxPAgent{t} \phi$ say respectively that ``it is compatible with everything $\agent$ knows that $\phi$ is true for reason $t$'' and ``$t$ is a reason why $\phi$ is permitted for agent $i$.''



For a formula $\phi$, the
set of all subformulas of $\phi$, denoted by $\Subf(\phi)$, is defined inductively as follows:
$\Subf(p)=\{ p \}$, for $p \in \Prop$;
$\Subf(\bot)=\{\bot\}$; $\Subf(* \phi)=\{* \phi\}\cup \Subf(\phi)$, where $* \in \{ \lnext, \lwprevious \}$; $\Subf(\phi \star \psi)=\{\phi\star \psi\}\cup \Subf(\phi)\cup
\Subf(\psi)$, where $\star \in \{ \limplies, \luntil, \lsince \}$; $\Subf(\jboxAgent{t} \phi)=\{\jboxAgent{t} \phi\}\cup \Subf(\phi)$; $\Subf(\jboxOAgent{t} \phi)=\{\jboxOAgent{t} \phi\}\cup \Subf(\phi)$. 

By expressing normative reasons explicitly in the language of $\JTO$, one can distinguish between those obligations that seem to conflict with each other. For example, there is an explicit conflict between the following obligations if they are expressed in the standard deontic logic:
\begin{itemize}
	\item It is obligatory for me that I meet Mary.
	$$
	\lobligatory_{{\sf me}} {\sf MeetMary}.
	$$
	
	\item It is obligatory for me that it is not the case that I meet Mary.
	$$
	\lobligatory_{{\sf me}} \neg {\sf MeetMary}.
	$$
	
\end{itemize}
However, there is no conflict of obligations if the reasons of obligations are explicitly mentioned:
\begin{itemize}
	\item It is obligatory for me that I meet Mary because I promised her.
	$$
	\jboxO{{\sf MyPromiseHer}}{{\sf me}} {\sf MeetMary}.
	$$
	
	\item It is obligatory for me that it is not the case that I meet Mary because I promised my wife.
	$$
	\jboxO{{\sf MyPromiseWife}}{{\sf me}} \neg {\sf MeetMary}.
	$$
	
\end{itemize}
In the above formulas suppose that ${\sf MyPromiseHer}, {\sf MyPromiseWife} \in \VTerms$.

%
%
%
%
%
%
%
%

Moreover, the combined language of justification logic and temporal logic allows for expressing some properties of systems that are not expressible in the known modal logics of knowledge, obligation and time. We give some examples here.

\begin{itemize}
	\setlength\itemsep{0.1cm}
	\item `Agent $i$ knows for reason $t$ that it is obligatory (or it is permitted) for her that $\phi$ is true for reason $s$' can be expressed by $\jboxAgent{t} \jboxOAgent{s} \phi$ (or $\jboxAgent{t} \jboxPAgent{s} \phi$).
	
	\item `It is permitted for agent $i$ that $\phi$ is true because of $t$ until she knows that $\psi$ holds because of $s$' can be expressed by $(\jboxPAgent{t} \phi) \luntil \jboxAgent{s} \psi$.  
	
	\item `It is obligatory for agent $i$ that $\phi$ is true because of $t$ since she knows that $\psi$ holds because of $s$' can be expressed by $(\jboxOAgent{t} \phi) \lsince \jboxAgent{s} \psi$.
	
	\item `If $t$ is agent $i$'s reason of why $\phi$ is obligatory at some past time, then $t$ is still her reason of why $\phi$ is obligatory (now)' can be expressed by $\lonce \jboxOAgent{t} \phi \limplies \jboxOAgent{t} \phi$. 
	
	\item `$t$ is agent $i$'s \textit{permanent reason} (or \textit{conclusive evidence}) of why $\phi$ is obligatory' can be expressed by $\lalways \jboxOAgent{t} \phi$ or, in a stronger from, by $\lalwaysPastFuture \jboxOAgent{t} \phi$.
	
	\item `Agent $i$ will have not forgotten her reason $t$ of why $\phi$ is obligatory by tomorrow, providing she possesses the reason now' can be expressed by $\jboxOAgent{t} \phi \limplies \lnext \jboxOAgent{t}$. 
	
	\item `Agent $i$ will learn that $t$ is a reason of why $\phi$ is obligatory  tomorrow, but she does not have this obligation now' can be expressed by $\neg \jboxOAgent{t} \phi \wedge \lnext \jboxOAgent{t}  \phi$. 
\end{itemize}

\section{Axioms}
\label{sec:Axioms}

The axiom system for temporal justification logic consists of four parts, namely propositional logic, temporal logic, epistemic justification logic, and deontic justification logic. For the temporal part, we use a system of~\cite{Gabbay,Gol87,Gor99} and  \cite{PnueliLichtensteinZuck1985,PnueliLichtenstein2000,vdMFrenchReynolds2005}.
For the epistemic justification logic part, we use a multi-agent version of the Logic of Proofs \cite{BucKuzStu11JANCL,Ghari2014,Ghari-IGPL-2021,TYav08TOCS}, and for the deontic justification logic part, we use a multi-agent version of the logic $\Logic{JNoC}^-$ (extended with the $\shiftrefax$ axiom) from \cite{Ghari-DEON-2021}. 

The axioms of the logic are:
\begin{enumerate}
	\setcounter{enumi}{\theenumsave}
	\item all propositional tautologies in the language of $\JTO$ \hfill \propax
	\setcounter{enumsave}{\theenumi}
\end{enumerate}

\noindent {\bf Axioms for the future operators:}
\begin{enumerate}
	\setcounter{enumi}{\theenumsave}
	\item $\lnext( \phi \limplies \psi) \limplies (\lnext \phi \limplies \lnext \psi)$ \hfill \nextkax
	\item $\lalways( \phi \limplies \psi) \limplies (\lalways \phi \limplies \lalways \psi)$ \hfill \alwayskax
	\item $\lnext \lneg \phi \liff \lneg \lnext \phi$ \hfill \funax
	\item $\lalways (\phi \limplies \lnext \phi) \limplies (\phi \limplies \lalways \phi)$ \hfill \indax
	\item $\phi \luntil \psi \limplies \leventually \psi$ \hfill \uoneax
	\item $\phi \luntil \psi \liff \psi \lor (\phi \land \lnext(\phi \luntil \psi))$ \hfill \utwoax
	\setcounter{enumsave}{\theenumi}
\end{enumerate}
\noindent {\bf Axioms for the past modalities:}
\begin{enumerate}
	\setcounter{enumi}{\theenumsave}
	\item $\lsofar( \phi \limplies \psi) \limplies (\lsofar \phi \limplies \lsofar \psi)$ \hfill \sofarkax
	\item $\lwprevious( \phi \limplies \psi) \limplies (\lwprevious \phi \limplies \lwprevious \psi)$ \hfill \prevkax
	\item $\lsprevious \phi \limplies \lwprevious\phi$ \hfill \swprevax
	\item $\lonce \lwprevious \bot$ \hfill \initialax
	\item $\lsofar (\phi \limplies \lwprevious \phi) \limplies (\phi \limplies \lsofar \phi)$ \hfill \sofarindax
	\item $\phi \lsince \psi \limplies \lonce \psi$ \hfill \soneax
	\item $\phi \lsince \psi \liff \psi \lor (\phi \land \lsprevious(\phi \lsince \psi))$ \hfill \stwoax
	\setcounter{enumsave}{\theenumi}
\end{enumerate}
\noindent {\bf Axioms for the interaction of future and past modalities:}
\begin{enumerate}
	\setcounter{enumi}{\theenumsave}
	
	\item $\phi \limplies \lnext \lsprevious \phi$ \hfill \fpax
	\item $\phi \limplies \lwprevious \lnext \phi$ \hfill \pfax
	
	\setcounter{enumsave}{\theenumi}
\end{enumerate}
%
\noindent {\bf Axioms for epistemic reasons:}
\begin{enumerate}
	\setcounter{enumi}{\theenumsave}
	\item $\jbox{t}_\agent (\phi \limplies \psi) \limplies (\jbox{s}_\agent \phi \limplies \jbox{t \tapp s}_\agent \psi)$ \hfill \appax
	
	\item $\jbox{t}_\agent \phi \rightarrow \jbox{t + s}_\agent \phi$, $\jbox{s}_\agent \phi \limplies  \jbox{t + s}_\agent \phi$ \hfill \sumax
	
	\item $\jbox{t}_\agent \phi \limplies \phi$ \hfill \refax
	
	\item $\jbox{t}_\agent \phi \limplies \jbox{\tinspect t}_\agent \jbox{t}_\agent \phi$ \hfill \posintax
	\setcounter{enumsave}{\theenumi}
\end{enumerate}
\noindent {\bf Axioms for normative reasons:}
\begin{enumerate}
	\setcounter{enumi}{\theenumsave}
	\item $\jboxOAgent{t} (\phi \limplies \psi) \limplies (\jboxOAgent{s} \phi \limplies \jboxOAgent{t \cdot s} \psi)$ \hfill \appOax
	
%

	\item $\jboxOAgent{t} \phi \limplies \jboxPAgent{t} \phi$ \hfill \nocax
	
	\item $\jboxOAgent{\tref t} (\jboxOAgent{t} \phi \limplies \phi)$ \hfill \shiftrefax
	
	\setcounter{enumsave}{\theenumi}
\end{enumerate}
The rules of inference are:
\[
\lrule{\vdash \phi \qquad \vdash \phi \limplies \psi}{\vdash \psi}\,\mprule,
\]
\[
\lrule{\vdash \phi}{\vdash \lnext \phi}\,\nextnecrule \, , \qquad \lrule{\vdash \phi}{\vdash \lwprevious \phi}\, \prevnecrule \, , \qquad \lrule{\vdash \phi}{\vdash \lalways\phi}\,\alwaysnecrule \, ,\qquad \lrule{\vdash \phi}{\vdash \lsofar\phi}\,\sofarnecrule,
\]

\[
\lrule{\jbox{c_{j_n}}_{i_n}\ldots\jbox{c_{j_1}}_{i_1} \phi \in \CS}{\vdash \jbox{c_{j_n}}_{i_n}\ldots\jbox{c_{j_1}}_{i_1} \phi}\ \iteratedconstnecrule,
\qquad
\lrule{\jboxO{c_{j_n}}{{i_n}}\ldots\jboxO{c_{j_1}}{{i_1}}  \phi \in \CS}{\vdash \jboxO{c_{j_n}}{{i_n}}\ldots\jboxO{c_{j_1}}{{i_1}}  \phi}\ \iteratedconstnecruleO.
\]
where in the iterated axiom necessitation rules $\iteratedconstnecrule$ and $\iteratedconstnecruleO$ the \textit{constant specification} $\CS$ is a set of formulas of the form
$$
\jbox{c_{j_n}}_{i_n}\ldots\jbox{c_{j_1}}_{i_1} \phi \quad \text{ or } \quad \jboxO{c_{j_n}}{{i_n}} \ldots \jboxO{c_{j_1}}{{i_1}} \phi,
$$ 
where $n \geq 1$, $i_1,\ldots,i_n$ are arbitrary agents,  $c_{j_n},\ldots,c_{j_1}$ are justification constants, and  $\phi$ is an axiom instance of  propositional logic, temporal logic, or justification logic (i.e. instances of axioms 1--23 above). Moreover, a constant specification $\CS$ should be downward closed in the sense that whenever $\jbox{c_{j_n}}_{i_n}\jbox{c_{j_{n-1}}}_{i_{n-1}}\ldots\jbox{c_{j_1}}_{i_1} \phi \in \CS$, where $n >1$, then $\jbox{c_{j_{n-1}}}_{i_{n-1}}\ldots\jbox{c_{j_1}}_{i_1} \phi \in \CS$, and whenever $\jboxO{c_{j_n}}{{i_n}} \jboxO{c_{j_{n-1}}}{{i_{n-1}}} \ldots \jboxO{c_{j_1}}{{i_1}}  \phi \in \CS$, where $n >1$, then $\jboxO{c_{j_{n-1}}}{{i_{n-1}}} \ldots \jboxO{c_{j_1}}{{i_1}}  \phi \in \CS$.

Given a constant specification $\CS$, we define the subsets $\CS^\lepistemic$ and $\CS^\lobligatory$ of $\CS$ as follows:
\begin{eqnarray*}
	\CS^\lepistemic &:= \lset{ \jbox{c_{j_n}}_{i_n}\ldots\jbox{c_{j_1}}_{i_1} \mid \jbox{c_{j_n}}_{i_n}\ldots\jbox{c_{j_1}}_{i_1} \in \CS}, \\
		\CS^\lobligatory &:= \lset{ \jboxO{c_{j_n}}{{i_n}} \ldots \jboxO{c_{j_1}}{{i_1}} \phi \mid \jboxO{c_{j_n}}{{i_n}} \ldots \jboxO{c_{j_1}}{{i_1}} \phi \in \CS}.
\end{eqnarray*}



For a given constant specification $\CS$, we shall use $\JTO_\CS$ to refer to the Hilbert system given by the axioms and rules for propositional logic, temporal logic, and justification logic (including both epistemic and normative reasons) as presented above. From here on when we write $\JTO_\CS$ we will assume that $\CS$ is a constant specification for $\JTO$. By $\vdash_\CS$ (or simply $\vdash$), we denote derivability in $\JTO_\CS$.

\paragraph{\textbf{Axiom $\nocax$}.}

The axiom $\nocax$, $\jboxOAgent{t} \phi \limplies \jboxPAgent{t} \phi$, expresses that ``if $\phi$ is obligatory for agent $\agent$ for a reason $t$, then $\phi$ is permitted for the agent for the same reason." This axiom is a justification version of the axiom $\lobligatory \phi \limplies \lpermissible \phi$ in standard deontic logic. It is worth noting that $\nocax$ is equivalent to
\[
\neg (\jboxOAgent{t} \phi \land \jboxOAgent{t} \neg \phi).
\]
The latter formula says that ``it is not the case that both $\phi$ and $\neg \phi$ are obligatory for agent $\agent$ for the same reason $t$," and hence for a normative reason there is no conflicts in norms.

There are other known axioms in the deontic justification logic context:
\begin{itemize}
	\item $\jboxOAgent{s} \phi \limplies \jboxPAgent{t} \phi$ \hfill \nocstrongax
	
	\item $\neg \jboxOAgent{t} \bot$ \hfill \conax 

\end{itemize}

The principle $\nocstrongax$ is used in \cite{Faroldi-Protopopescu-IGPL-2019} as an axiom in a justification logic of normative reasons. This principle says that ``if $\phi$ is obligatory for agent $\agent$ for reason $s$, then it is permitted for agent $\agent$ for (maybe another) reason $t$.'' In this form, it seems that this principle is not plausible. For example, consider the following instance: If $\agent$ is obliged to return the book to the library because of the library's law, then she is permitted to return the book because she is hungry (!). It is worth noting that $\nocstrongax$ is equivalent to
$$
\neg (\jboxOAgent{s} \phi \land \jboxOAgent{t} \neg \phi),
$$
which says that ``it is not the case that $\phi$ is obligatory for agent $\agent$ for a reason $s$ while $\neg \phi$ is obligatory for agent $\agent$ for a reason $t$." This is a strong version of axiom $\nocax$.

The principle $\conax$ is usually used in justification logics as a justification version of axiom $\neg \lobligatory \bot$ in standard deontic logic (see, for instance \cite{Fit16APAL,Ghari-APAL-2017,KuzStu12AiML}). This principle says that ``it is not obligatory for agent $i$ that contradiction is true for reason $t$.'' Using this axiom, one can prove $\nocstrongax$  (cf. \cite{Ghari-DEON-2021} for a more detailed exposition), which is not a plausible principle as states above.

\paragraph{\textbf{Basic properties}.}

The definition of derivation from a set of assumptions is standard:
\[
T \vdash_\CS \phi \text{ iff there exist $\psi_1,\ldots,\psi_n \in T$ such that 
	$\vdash_\CS (\psi_1 \land \cdots \land \psi_n) \to \phi$.}
\]

From the definition of derivation from a set of assumptions, it follows that the Deduction Theorem holds in $\JTO_\CS$, i.e. $T \vdash_\CS \phi \limplies \psi$ if{f}  $T, \phi \vdash_\CS \psi$.

\begin{lemma}\label{lem: Admissible rules in LTL}
	The following rules are admissible in $\JTO_\CS$:
	\[
	\lrule{T \vdash \phi}{\lalwaysPastFuture T \vdash \lalwaysPastFuture \phi}\,\alwaysPastFuturenecrule, \qquad
	 \lrule{\vdash \phi \limplies \psi}{\vdash \lwprevious \phi \limplies \lwprevious \psi}\, \prevRMrule, \qquad
	  \lrule{\vdash \phi \limplies \psi}{\vdash \lonce \phi \limplies \lonce \psi}\, \onceRMrule.
	\]
\end{lemma}
\begin{proof}
	Straightforward. \qed 
\end{proof}

In the next lemma we state some basic results of linear temporal logic that we need later on.

\begin{lemma}\label{lem: basic results of linear temporal logic}
	The following formulas are provable in $\JTO_\CS$:
	
	\begin{enumerate}
		\item $\lalways \phi \limplies (\phi \land \lnext \lalways \phi)$.
		\item $\lalways \phi \limplies  \lnext \phi$.
		\item $\lsofar \phi \limplies (\phi \wedge \lwprevious \lsofar \phi)$.
		\item $\lsofar \phi \limplies  \lwprevious \phi$.
		\item $\lalwaysPastFuture \phi \limplies \phi$.
		\item $\lalwaysPastFuture \phi \liff \lalwaysPastFuture \lalwaysPastFuture \phi$.
		\item $\phi \lunless \psi \liff \psi \lor (\phi \land \lnext(\phi \lunless \psi))$.
	\end{enumerate}
	
\end{lemma}
\begin{proof}
	Straightforward. \qed 
\end{proof}

\section{Semantics}
\label{sec:Semantics Fiittng}

In this section, a semantics builds on top of interpreted systems and Fitting models, called F-interpreted systems, is introduced for $\JTO$. Semantics for the temporal part is based on linear and discrete time flow, associating a time point with any natural number. For the justification part, we use Fitting models which is first introduced for the logic of proofs (cf. \cite{Fit05APAL}). The semantics introduced in this section is similar to that given in \cite{BucheliGhariStuder2017,Ghari-IGPL-2021} for temporal justification logics.

\begin{definition}\label{def:frame-run-system}
	Given a non-empty set of states $S$, a \emph{run}~$r$ on $S$ is a function from $\N$ to $S$, i.e., $r: \N \to S$. A \emph{system}~$\runs$ on a set of states $S$ is a non-empty set of runs. 
	
\end{definition}

Given a run $r$ and $n \in \N$, the pair $(r,n)$ is called a \emph{point}. The image of $\runs$ is defined as follows:
\[
\Im(\runs) := \{ r(n) \mid r \in \runs , n \in \N \}.
\]
Note that $\emptyset \neq \Im(\runs) \subseteq S$, but $\Im(\runs)$ and $S$ are not necessarily equal. 

\begin{definition}\label{def: F-frame-run-system}
	A \emph{frame} is a tuple $(S, \runs, \accrel_1,\ldots,\accrel_\numberofagents, \accrelO_1,\ldots, \accrelO_\numberofagents)$ (or $(S, \runs, \accrel_\agent, \accrelO_\agent)_{\agent \in \Ag}$ for short) where
	\begin{enumerate}
		\item $S$ is a non-empty set of states;
		
		\item $\runs$ is a system on $S$;
		
		\item each $\accrel_\agent \subseteq S \times S$ is a reflexive and transitive relation; 
		
		\item each $\accrelO_\agent \subseteq S \times S$ is a shift reflexive relation.\footnote{A relation $\accrel$ is shift reflexive if  $w \accrel v$ implies $v \accrel v$, for every states $w,v$.} 
	\end{enumerate}
	%

\end{definition}

\begin{definition}\label{def:evidence function}
	Given a frame $(S, \runs, R_\agent, R^\lobligatory_\agent)_{\agent \in \Ag}$,
	\emph{evidence functions and normative evidence functions for agent $\agent$} are functions $\evidence_\agent: S \times \Terms^\lepistemic \to \powerset(\Formulae)$ and $\evidenceO_\agent: S \times \Terms^\lobligatory \to \powerset(\Formulae)$, respectively, satisfying the following conditions.
	
	\noindent
	\textbf{Conditions on $\evidence_\agent$:}\\ For all terms $s,t \in \Terms^\lepistemic$, all formulas $\phi,\psi \in \Formulae$, all states $v,w \in S$, and all $i \in \Ag$:
	\begin{enumerate}
		\setlength\itemsep{0.01cm}
		\item 
		$\evidence_\agent(v,t) \subseteq \evidence_\agent(w,t)$, whenever $v \accrel_\agent w$; \hfill (monotonicity)
		
		\item 
		if $\jboxAgent{t} \phi \in \CS^\lepistemic$, then $\phi \in \evidence_\agent(w,t)$; \hfill (constant specification)
		
		\item 
		if $\phi \limplies \psi \in \evidence_\agent(w,t)$ and $\phi \in \evidence_\agent(w,s)$, then $\psi \in \evidence_\agent(w, t \cdot s)$; \hfill (application)
		
		
		%

		\item 
		if $\phi \in \evidence_\agent(w,t)$, then $\jbox{t}_\agent \phi \in \evidence_\agent(w,\tinspect t)$. \hfill (positive introspection)
	\end{enumerate}
	\textbf{Conditions on $\evidenceO_\agent$:} \\ For all terms $s,t \in \Terms^\lobligatory$, all formulas $\phi,\psi \in \Formulae$, all states $v,w \in S$, and all $i \in \Ag$:
	\begin{enumerate}
		\setlength\itemsep{0.01cm}
		
		\item 
		if $\jboxOAgent{t} \phi \in \CS^\lobligatory$, then $\phi \in \evidenceO_\agent(w,t)$; \hfill (constant specification-$\lobligatory$)
		
		\item 
		if $\phi \limplies \psi \in \evidenceO_\agent(w,t)$ and $\phi \in \evidenceO_\agent(w,s)$, then $\psi \in \evidenceO_\agent(w, t \cdot s)$; \hfill (application-$\lobligatory$)
		
		
		%
		\item 
		If $\phi \in \evidenceO_\agent(w,t)$, then $\neg \phi \not \in \evidenceO_\agent(w,t)$; \hfill (consistency)
		
		\item 
		$\jboxOAgent{t} \phi \limplies \phi \in \evidenceO_\agent(w,\tref t)$; \hfill (obligated factivity)
		
	\end{enumerate}
	
\end{definition}

Intuitively, $\phi \in \evidence_\agent(w,t)$ is read ``$t$ is a reason why $\phi$ is known for agent $i$ in state $w$,'' and $\phi \in \evidenceO_\agent(w,t)$ is read ``$t$ is a reason why $\phi$ is obligatory for agent $i$ in state $w$.''

It is worth noting that, without loss of generality, one may define the relations $\accrel_\agent$ and $\accrelO_\agent$ as binary relations on $\Im(\runs)$ and then restrict  the domain of evidence functions to $\Im(\runs) \times \Terms$, and state the above conditions on $\evidence_\agent$ and $\evidenceO_\agent$ only for states in $\Im(\runs)$.

\begin{definition}\label{def:quasi-interpreted sysytems}
	An \emph{F-interpreted system for $\JTO_\CS$} is a tuple 
	\[
	\system = (S, \runs,  \accrel_\agent, \evidence_\agent, \accrelO_\agent, \evidenceO_\agent,  \valuation)_{\agent \in \Ag}
	\] 
	where
	\begin{enumerate}
		\item $(S, \runs, R_\agent, R^\lobligatory_\agent)_{\agent \in \Ag}$ is a frame;
		\item $\evidence_\agent$ is an evidence function for agent~$\agent \in \Ag$;
		\item $\evidenceO_\agent$ is a normative evidence function for agent~$\agent \in \Ag$;
		\item $\valuation: \Im(\runs) \to \powerset(\Prop)$ is a valuation.
	\end{enumerate}
\end{definition}

\begin{definition}\label{def:truth conditions interpreted systems}
	Given an F-interpreted system  $\system = (S, \runs,  \accrel_\agent, \evidence_\agent, \accrelO_\agent, \evidenceO_\agent,  \valuation)_{\agent \in \Ag}$ a run $r \in \runs$, and  $n \in \N$, we define truth of a formula $\phi$ in $\system$ at point $(r,n)$ inductively as follows: 
	\begin{align*}
		(\system, r, n) &\entails p \text{ iff } p \in \valuation(r(n)), \text{ for } p \in \Prop \, ,\\
		(\system, r, n) &\not\entails \lfalse \, ,\\
		(\system, r, n) &\entails \phi \limplies \psi \text{ iff } (\system, r, n) \not\entails \phi \text{ or } (\system, r, n) \entails \psi \, ,\\
		(\system, r, n) &\entails \lwprevious \phi \text{ iff $n=0$ or } (\system, r, n-1) \entails \phi \, ,\\
		(\system, r, n) &\entails \lnext \phi \text{ iff } (\system, r, n+1) \entails \phi \, ,\\
		(\system, r, n) &\entails \phi \lsince \psi \text{ iff there is some } m \leq n \text{ such that } (\system, r, m) \entails \psi \\ & \qquad\qquad \text{ and } (\system, r, k) \entails \phi \text{ for all $k$ with } m < k \leq n \, ,\\
		(\system, r, n) &\entails \phi \luntil \psi \text{ iff there is some } m \geq n \text{ such that } (\system, r, m) \entails \psi \\ & \qquad\qquad \text{ and } (\system, r, k) \entails \phi \text{ for all $k$ with  } n \leq k < m \, ,\\ 
		(\system, r, n) &\entails \jboxAgent{t} \phi \text{ iff }  \phi \in \evidence_\agent(r(n),t)  \text { and }  (\system, r^\prime, n^\prime) \entails \phi \\ &\qquad\qquad \text{ for all } r^\prime \in \runs \text{ and } n^\prime \in \N \text{ such that }   r(n) R_\agent r'(n')\, ,\\   
		(\system, r, n) &\entails \jboxOAgent{t} \phi \text{ iff }  \phi \in \evidenceO_\agent(r(n),t)  \text { and }  (\system, r^\prime, n^\prime) \entails \phi \\ &\qquad\qquad \text{ for all } r^\prime \in \runs \text{ and } n^\prime \in \N \text{ such that }   r(n) R^\lobligatory_\agent r'(n').
	\end{align*}
\end{definition}

From the above definitions it follows that:
\begin{align*}
	(\system, r, n) &\entails \leventually \phi \text{ iff } (\system, r, m) \entails \phi \text{ for some } m\geq n \, ,\\
	(\system, r, n) &\entails \lalways \phi \text{ iff } (\system, r, m) \entails \phi \text{ for all } m\geq n \, ,\\
	(\system, r, n) &\entails \lonce \phi \text{ iff } (\system, r, m) \entails \phi \text{ for some } m \leq n \, ,\\
	(\system, r, n) &\entails \lsofar \phi \text{ iff } (\system, r, m) \entails \phi \text{ for all } m \leq n \, ,\\
	(\system, r, n) &\entails \lsprevious \phi \text{ iff $n>0$ and } (\system, r, n-1) \entails \phi \, \\
	(\system, r, n) &\entails \lalwaysPastFuture \phi \text{ iff } (\system, r, m) \entails \phi \text{ for all } m \, .
\end{align*}
Define the truth set of $\phi$ in the F-interpreted system $\system$ as follows:
$$\truthset{\phi}{\system} = \{ r(n) \in \Im(\runs) \mid  (\system, r,n) \entails \phi \}.$$
Let 
\[
\accrel_\agent(r(n)) = \{ r'(n') \in \Im(\runs) \mid  r(n) \accrel_\agent r'(n') \}, 
\]
\[
 \accrelO_\agent(r(n)) = \{ r'(n') \in \Im(\runs) \mid  r(n) \accrelO_\agent r'(n') \}.
\]
Now the truth condition of justification assertions can be expressed as follows:
\begin{align*}
	(\system, r, n) &\entails \jboxAgent{t} \phi \text{ iff }  \phi \in \evidence_\agent(r(n),t)  \text { and }  \accrel_\agent(r(n)) \subseteq \truthset{\phi}{\system}\, ,\\   
	(\system, r, n) &\entails \jboxOAgent{t} \phi \text{ iff }  \phi \in \evidenceO_\agent(r(n),t)  \text { and }  \accrelO_\agent(r(n)) \subseteq \truthset{\phi}{\system}.
\end{align*} 

\begin{definition}
	Let $\CS$ be a constant specification for $\JTO$.
	\begin{enumerate}
		\item Given an F-interpreted system $\system = (S, \runs, \ldots)$ for $\JTO_\CS$, we write $\system \entails \phi$ if
		for all $r \in \runs$ and all $ n \in \N$, we have 
		$(\system, r, n) \entails \phi$.
		
		\item We write $\entails_{\CS} \phi$ if $\system \entails \phi$ for all 
		interpreted systems $\system$ for $\JTO_\CS$.
		
		\item Given a set of formulas $T$ and a formula $\phi$ of $\JTO_\CS$, the (local) consequence relation is defined as follows: $T \models_{\CS} \phi$ iff for all 
		interpreted systems $\system = (S, \runs, \ldots)$ for $\JTO_\CS$, for all $r \in \runs$, and for all $ n \in \N$, if $(\system, r, n) \entails \psi$ for all $\psi \in T$, then $(\system, r, n) \entails \phi$.
	\end{enumerate}
\end{definition}

\paragraph{\textbf{Expressing moments of time}.}

In the rest of this section we show how linear temporal logic with past operators can express moments of time in the language. All the results of this part can be obtained by purely temporal reasoning in linear temporal logic.

From the definition of truth in interpreted systems we have
\[
(\system, r, n) \entails \lwprevious \bot \text{ iff $n=0$. }
\]
Thus, $\lwprevious \bot$ expresses the property `the time is 0.' Let $\lsprevious^m$ denote the  $m$ times iteration of $\lsprevious$. Then,  $\lsprevious^m \lwprevious\bot$ expresses the property `the time is m.' 

Let `$\ltime=m$' abbreviate $\lsprevious^m \lwprevious \bot$ and let $\ltrue_m (\phi)$ abbreviate 
$$
\lalwaysPastFuture (\ltime = m \limplies \phi).
$$
Note that the natural number $m$ in the formulas $\ltime=m$ and $\ltrue_m (\phi)$ is not part of the language. This number only shows the number of iterations of the modality $\lsprevious$. The following lemma shows that the formula $\ltrue_m (\phi)$ expresses that `$\phi$ is true at time $m$.' We call $\ltrue$ the \textit{temporal truth predicate}.

\begin{lemma}
	Given an F-interpreted system $\system = (S, \runs, \ldots)$ for $\JTO_\CS$, for all $r \in \runs$ and for all $n, m \in \N$ we have:
	\begin{enumerate}
		\item $(\system, r, n) \entails \ltime=m$  iff $n=m$.
		
		\item $(\system, r, n) \entails \ltrue_m (\phi)$ iff $(\system, r, m) \entails \phi$.
		
		\item $\system \models \ltrue_m (\phi)$ iff $(\system, r, m) \models \phi$  for all $r \in \mathcal{R}$.
	\end{enumerate}
\end{lemma}
\begin{proof}
	The proof of item 1 is by induction on $m$. Item 2 follows easily from 1, and item 3 follows easily from 2. \qed
\end{proof}

The temporal truth predicate satisfies the following intuitive properties.

\begin{lemma}\label{lem: temporal truth predicate properties}
	The following formulas are valid in $\JTO_\CS$. For all $m \in \N$: 
	\begin{enumerate}
		\item $\entails_\CS (\ltrue_m(\phi) \wedge \ltime = m) \limplies \phi$.
				
		\item $\entails_\CS \ltrue_m(\lnext \phi) \liff \ltrue_{m+1} (\phi)$.
		
		\item $\entails_\CS \ltrue_{m+1}(\lsprevious \phi) \liff \ltrue_{m} (\phi)$.
		
		\item $\entails_\CS \ltrue_{m+1}(\lwprevious \phi) \liff \ltrue_{m} (\phi)$.
		
		\item $\entails_\CS \ltrue_{m+1}(\lsofar \phi) \limplies \ltrue_{m}(\phi)$.
		
		\item $\entails_\CS \lalwaysPastFuture \phi \limplies \ltrue_m(\phi)$.
		
		\item $\entails_\CS \lalwaysPastFuture \ltrue_m(\phi) \liff \ltrue_m(\phi)$.
		
		\item $\ltime = m \limplies \phi \entails_\CS \ltime = m+1 \limplies \lsprevious \phi$. 
		
		\item If $\entails_\CS \phi \limplies \psi$, then $\entails_\CS \ltrue_m (\phi) \limplies \ltrue_m (\psi)$.
		
	\end{enumerate}
\end{lemma}
\begin{proof}
	Straightforward. \qed
\end{proof}

By making use of temporal modalities $\lsince$ and $\luntil$ and the temporal truth predicate, it is not hard to express that a formula is true in a time interval (for an interval-based deontic logic see \cite{Glavanicova-Pascucci-DEON-2021}). Let $m,n \in \N$ and $m < n$.
\begin{itemize}
	\setlength\itemsep{0.01cm}
	\item `$\phi$ is always true until the time is $m$ (excluding $m$)'  can be formalized by $\phi \luntil \ltime = m$. This is denoted by $\lalways_{[now, m)} \phi$. Likewise, the formula $\phi \luntil \ltime = m+1$, denoted by $\lalways_{[now, m]} \phi$, expresses that `$\phi$ is always true until the time is $m$ (including $m$)'.
	
	
	\item `$\phi$ has been always true since the time was $m$ (excluding $m$)'  can be formalized by $\phi \lsince \ltime = m$. This is denoted by $\lalways_{(m,now]} \phi$. Likewise, the formula $\phi \lsince \ltime = m-1$, denoted by $\lalways_{[m, now]} \phi$, expresses that `$\phi$ has been always true since the time was $m$ (including $m$)'
	
	
	\item `$\phi$ is always true in the time interval $[m,n]$' can be formalized by $\ltrue_m (\phi \luntil \ltime = n+1)$. This is denoted by $\lalways_{[m, n]} \phi$. In this case, the formula $\ltrue_{n} (\phi \lsince \ltime = m-1)$ or $\lalwaysPastFuture ( (\ltime =m \vee \ltime = m+1 \vee \cdots \vee \ltime = n) \limplies \phi)$ can also formalize this property. Likewise, we have $\lalways_{[m, n)} \phi := \ltrue_m (\phi \luntil \ltime = n)$, $\lalways_{(m, n]} \phi := \ltrue_{n} (\phi \lsince \ltime = m)$ , $\lalways_{(m, n)} \phi := \ltrue_{n-1} (\phi \lsince \ltime = m)$.
	
	\item `$\phi$ is sometimes true in the time interval $[m,n]$' can be formalized by $\leventually (\phi \wedge (\ltime =m \vee \ltime = m+1 \vee \cdots \vee \ltime = n) ) \vee \lonce (\phi \wedge (\ltime =m \vee \ltime = m+1 \vee \cdots \vee \ltime = n) )$. This is denoted by $\leventually_{[m, n]} \phi$. The formulas $\leventually_{[m, n)} \phi$, $\leventually_{(m, n]} \phi$, and $\leventually_{(m, n)} \phi$ can be formalized in a similar manner.
\end{itemize}

Using the temporal truth predicate one can express normative sentences more precisely. For example, there is an explicit conflict between the following obligations:
\begin{itemize}
	\item It is obligatory for me that I meet Mary  because I promised her.
	$$
	\jboxO{ {\sf MyPromiseHer} }{{\sf me}} {\sf MeetMary}.
	$$
	
	\item It is obligatory for me that it is not the case that I meet Mary because I promised her.
	$$
	\jboxO{ {\sf MyPromiseHer} }{{\sf me}} \neg {\sf MeetMary}.
	$$
	
\end{itemize}
However, there is no conflict of obligations if the time of the appointment is explicitly mentioned:
\begin{itemize}
	\item It is obligatory for me that I meet Mary on Sunday (here denoted by 0) because I promised her.
	$$
	\jboxO{{\sf MyPromiseHer}}{{\sf me}} \ltrue_{0}({\sf MeetMary}).
	$$
	
	\item It is obligatory for me that it is not the case that I meet Mary on Monday (here denoted by 1) because I promised her.
	$$
	\jboxO{{\sf MyPromiseHer}}{{\sf me}} \ltrue_{1}(\neg {\sf MeetMary}).
	$$
\end{itemize}
In the above formulas suppose that ${\sf MyPromiseHer} \in \VTerms$.

\section{Completeness}
\label{sec:Completeness}

In this section we show soundness and completeness of $\JTO_\CS$ with respect to F-interpreted systems.


For a formula $\chi$, let 
\begin{gather*}
	A_\chi \colonequals \Subf(\chi) \cup  \Subf(\top \lsince \lwprevious \bot),  
	\\[0.1cm]
	\Subf^+(\chi) \colonequals A_\chi \cup \{ \neg \psi \ |\  \psi \in A_\chi \}.
\end{gather*}
%
%
It is easy to show that for a formula $\chi$, the sets $\Subf(\chi)$ and $\Subf^+(\chi)$ are finite.

\begin{definition}
	Let $\CS$ be a constant specification for $\JTO$.
	\begin{itemize}
		\item A set $\Gamma$ of formulas is called \textit{consistent} if $\Gamma \not\vdash_\CS \bot$.
		
		\item A set $\Gamma$ of formulas is called \textit{maximal} if it has no consistent proper extension of formulas.
		
		\item A set $\Gamma \subseteq  \Subf^+(\chi)$  is called \textit{$\chi$-maximal} if it has no consistent proper extension of formulas from $ \Subf^+(\chi)$.
	\end{itemize}
	
\end{definition} 

Consistent sets can be extended to maximal consistent sets by the Lindenbaum Lemma.

Let $\MCS_\chi$ denote the set of all $\chi$-maximally consistent subsets of $ \Subf^+(\chi)$. Note that $\MCS_\chi$ is a finite set. Let $\MCS$ denote the set of all maximally consistent sets, and for $\Gamma\in \MCS$ let $$\overline{\Gamma} \colonequals \Gamma \cap  \Subf^+(\chi).$$ 
%

\begin{lemma}\label{lem:characterization of MCS-chi}
	For every formula $\chi$ we have:
	\[
	\MCS_\chi = \{\Gamma \cap  \Subf^+(\chi) \mid \Gamma\in\MCS \}.
	\]
\end{lemma}
\begin{proof}
The proof is similar to that given in \cite{Ghari-IGPL-2021}. \qed
%
%
\end{proof}

\begin{definition}
	The relation $R_\lnext$ on $\MCS_\chi$ is defined as follows:
	\[
	X R_\lnext Y 
	\text{ if{f} }
	\text{there exists $\Gamma, \Delta \in \MCS$ such that $X = \overline{\Gamma}$ and $Y = \overline{\Delta}$ and } \{ \phi \ |\  \lnext \phi \in \Gamma\} \subseteq \Delta.
	\]
	The notation $\RO{X}{Y}{\overline{\Gamma}}{\overline{\Delta}}$ means that $X, Y \in \MCS_\chi$, $\Gamma, \Delta \in \MCS$, $X = \overline{\Gamma}$, $Y = \overline{\Delta}$ and $\{ \phi \ |\  \lnext \phi \in \Gamma\} \subseteq \Delta$. Note that for $\Gamma, \Delta \in \MCS$ if $\{ \phi \ |\  \lnext \phi \in \Gamma\} \subseteq \Delta$, then $\overline{\Gamma} R_\lnext \overline{\Delta}$.
\end{definition}

\begin{definition}\label{Def:acceptable sequence}
	
	An infinite sequence $(X_0, X_1, \ldots)$ of elements of $\MCS_\chi$ is called \textit{acceptable} if 
	\begin{enumerate}
		\item $X_n R_\lnext X_{n+1}$ for all $n \geq 0$, and
		\item for all $n$, if $\phi \luntil \psi \in X_n$, then there exists $m \geq n$ such that $\psi \in X_m$ and $\phi \in X_k$ for all $k$ with $n \leq k <m$.
		\item $\lwprevious \bot \in X_0$. 
	\end{enumerate}
\end{definition}

Let $(X_0, X_1, \ldots)$ be an acceptable sequence of elements of $\MCS_\chi$. Then, there exists $\Gamma_{0}, \Gamma_{1}, \ldots \in \MCS$ and $\Delta_{1}, \Delta_{2}, \ldots \in \MCS$ such that
$$\RO{X_{0}}{X_{1}}{\overline{\Gamma_{0}}}{\overline{\Delta_{1}}}, \RO{X_{1}}{X_{2}}{\overline{\Gamma_{1}}}{\overline{\Delta_{2}}},\RO{X_{2}}{X_{3}}{\overline{\Gamma_{2}}}{\overline{\Delta_{3}}}, \ldots.$$

Note that for each $n \geq 0$ we have $\RO{X_{n}}{X_{n+1}}{\overline{\Gamma_{n}}}{\overline{\Delta_{n+1}}}$, and thus $X_0 = \overline{\Gamma_{0}}$ and $X_i = \overline{\Gamma_{i}} = \overline{\Delta_{i}}$ for all $i > 0$. 
%

The following results are proved in \cite{Ghari-IGPL-2021,HvdMV04}, so we omit the proofs here.

\begin{lemma}\label{lem:R-next}
	Let $\RO{X}{Y}{\overline{\Gamma}}{\overline{\Delta}}$.
	\begin{enumerate}
		\item  $\lnext \phi \in \Gamma$ iff $\phi \in \Delta$.
		
		\item $\phi \in \Gamma$ iff $\lsprevious \phi \in \Delta$.
		
		\item $\phi \in \Gamma$ iff $\lwprevious \phi \in \Delta$.
	\end{enumerate}
\end{lemma}

%

\begin{theorem}\label{thm: acceptable sequence}
	For every $X \in \MCS_\chi$, there is an acceptable sequence containing $X$. 
\end{theorem}

\begin{lemma}\label{lem: MCS temporal properties}
	Let $(X_0, X_1, \ldots)$ be an acceptable sequence of elements of $\MCS_\chi$, let $n \geq 0$, and let $\RO{X_{i-1}}{X_{i}}{\overline{\Gamma_{i-1}}}{\overline{\Delta_{i}}}$, for  $i \geq 1$.
		
		\begin{enumerate}
			\item If $\phi \lsince \psi \in X_n$, then there exists $m \leq n$ such that $\psi \in X_m$ and $\phi \in X_k$ for all $k$ with $m < k \leq n$.
			
			\item If $\phi \lsince \psi \in \Subf^+(\chi)$ and there exists $m \leq n$ such that $\psi \in X_m$ and $\phi \in X_k$ for all $k$ with $m < k \leq n$, then $\phi \lsince \psi \in X_n$.
		\end{enumerate}
		
%
%
%
%
%
%
%
%
%
	
\end{lemma}


For the completeness proof we employ the canonical model construction.

\begin{definition}\label{def:canonical interpreted systems for LPLTL}
	The $\chi$-canonical F-interpreted system 
	$
	\system = (S, \runs,  \accrel_\agent, \evidence_\agent, \accrelO_\agent, \evidenceO_\agent,  \valuation)_{\agent \in \Ag}
	$ 
	for $\JTO_\CS$ is defined as follows:
	\begin{enumerate}[itemsep=2pt,parsep=2pt]
		\item $S \colonequals \MCS_\chi$;
		
		\item $\runs$ consists of all mappings $r: \N \to \MCS_\chi$ such that
		$(r(0), r(1), \ldots)$ is an acceptable sequence;
		%
		
		%
		
		\item $X \accrel_\agent Y$ if{f} for all $\Delta \in \MCS$ such that $Y = \overline{\Delta}$ there exists $\Gamma \in \MCS$ such that $X = \overline{\Gamma}$ and  
		$\{ \phi \mid  \jbox{t}_\agent \phi \in \Gamma,   \text{ for some $t$}\} \subseteq \Delta$, for $X,Y \in \MCS_\chi$; 

		\item $\evidence_\agent (X, t) \colonequals \{ \phi \mid \jbox{t}_\agent \phi \in \bigcap \{ \Gamma \in \MCS \mid X = \overline{\Gamma} \} \}$, for $X \in \MCS_\chi$ and $t \in \Terms^\lepistemic$; 
		\item $X \accrelO_\agent Y$ if{f} there are $\Gamma, \Delta \in \MCS$ such that $X = \overline{\Gamma}$, $Y = \overline{\Delta}$, and 
		$\{ \phi \mid  \jboxO{t}{\agent} \phi \in \Gamma,   \text{ for some $t$}\} \subseteq \Delta$, for $X,Y \in \MCS_\chi$; 

		\item $\evidenceO_\agent (X, t) \colonequals \{ \phi \mid \jboxO{t}{\agent} \phi \in \bigcap \{ \Gamma \in \MCS \mid X = \overline{\Gamma} \} \}$, for $X \in \MCS_\chi$ and $t \in \Terms^\lobligatory$;
		\item $\valuation(X) \colonequals \Prop\cap X$, for $X \in \MCS_\chi$.
	\end{enumerate}
\end{definition}

Note that $\MCS_\chi = \{ r(n) \ |\ r\in \runs, n \in \N \}$, and hence $S = \Im(\runs)$. 

\begin{lemma}\label{lem: canonical interpreted model is a model}
	The $\chi$-canonical F-interpreted system 
	$
	\system = (S, \runs,  \accrel_\agent, \evidence_\agent, \accrelO_\agent, \evidenceO_\agent,  \valuation)_{\agent \in \Ag}
	$ 
	for $\JTO_\CS$ is an F-interpreted system for $\JTO_\CS$.
\end{lemma}
\begin{proof}
	It is not difficult to show that each $R_\agent$ is reflexive and transitive. Let us show that $\accrelO_\agent$ is shift reflexive. Suppose that $X \accrelO_\agent Y$. Then, there are $\Gamma, \Delta \in \MCS$ such that $X = \overline{\Gamma}$, $Y = \overline{\Delta}$, and 
	$\{ \phi \mid  \jboxO{t}{\agent} \phi \in \Gamma   \text{ for some $t$}\} \subseteq \Delta$. It suffices to show that $\{ \phi \mid  \jboxO{t}{\agent} \phi \in \Delta,  \text{ for some $t$}\} \subseteq \Delta$. Suppose that $\jboxOAgent{t} \phi \in \Delta$, for some $t \in \Terms^\lobligatory$.  Since $\jboxOAgent{\tref t} (\jboxOAgent{t} \phi \limplies \phi) \in \Gamma$, we get $\jboxOAgent{t} \phi \limplies \phi \in \Delta$, and hence $\phi \in \Delta$. 
	
	We now have to show that evidence functions satisfy the conditions of Definition \ref{def:evidence function}. We only show the $(monotonicity)$ condition for $\evidence_\agent$ and the $(consistency)$ condition for $\evidenceO_\agent$. The proof for other conditions are easy.
	
	For $(monotonicity)$, suppose that $X, Y \in S$, $X R_\agent Y$, and $\phi \in \evidence_\agent (X,t)$. We have to show that $\phi \in \evidence_\agent (Y,t)$. Let $Y = \overline{\Delta}$, for an arbitrary $\Delta \in \MCS$. We have to show that $\jbox{t}_\agent \phi \in \Delta$. Since $X R_\agent Y$, there exists $\Gamma \in \MCS$ such that $X = \overline{\Gamma}$ and  
	$\{ \phi \mid  \jbox{t}_\agent \phi \in \Gamma   \text{ for some $t$}\} \subseteq \Delta$. Since $\phi \in \evidence_\agent (X,t)$, we get $\jbox{t}_\agent \phi \in \Gamma$, and thus $\jbox{!t}_\agent \jbox{t}_\agent \phi \in \Gamma$. Therefore, $\jbox{t}_\agent \phi \in \Delta$, as desired.
	
	For $(consistency)$, suppose that $\phi \in \evidenceO_\agent (X,t)$, and assume to obtain a contradiction that  $\neg \phi \in \evidenceO_\agent (X,t)$. Thus, for every $\Gamma \in \MCS$ such that $X = \overline{\Gamma}$ we have $\jboxO{t}{\agent} \phi \in \Gamma$ and $\jboxO{t}{\agent} \neg \phi \in \Gamma$. This contradicts the fact that $\neg (\jboxO{t}{\agent} \phi \wedge \jboxO{t}{\agent} \neg \phi) \in \Gamma$. \qed
\end{proof}

\begin{lemma}[Truth Lemma]\label{lem:Truth Lemma F-interpreted}
	Let 
	$
	\system = (S, \runs,  \accrel_\agent, \evidence_\agent, \accrelO_\agent, \evidenceO_\agent,  \valuation)_{\agent \in \Ag}
	$ 
	be the $\chi$-canonical F-interpreted system for $\JTO_\CS$. For every formula $\psi \in  \Subf^+(\chi)$, every run~$r$ in $\runs$, and every $n \in \N$ we have:
	\[
	(\system, r, n) \models \psi 
	\quad\text{if{f}}\quad
	\psi \in r(n).
	\]
\end{lemma}
\begin{proof}
	As usual, the proof is by induction on the structure of $\psi$. Let $\RO{r(i-1)}{r(i)}{\overline{\Gamma_{i-1}}}{\overline{\Delta_{i}}}$, for $i \geq 1$. We only show the following cases:
	\begin{itemize}
		\item $\psi = \lnext \phi$.
		
		$(\system, r, n) \models \lnext \phi$ iff $(\system, r, n+1) \models \phi$, by the induction hypothesis, iff $\phi \in r(n+1)$ iff $\phi \in \Delta_{n+1}$, by Lemma \ref{lem:R-next}, iff $\lnext \phi \in \Gamma_{n}$ iff $\lnext \phi \in r(n)$.
		
		%
		%

		\item $\psi = \lwprevious \phi$.
		
		$(\Rightarrow)$ Suppose that $(\system, r, n) \models \lwprevious \phi$ and $\lwprevious \phi \not\in r(n)$.  Then $n=0$ or $(\system, r, n-1) \models  \phi$.
		
		\begin{itemize}
			\item Suppose $n=0$. Since $r(0)$ is initial, $\lwprevious \bot \in r(0)$. Since  $\vdash_\CS \lwprevious \bot \rightarrow \lwprevious \phi$, we get $\lwprevious\phi \in r(0)$. The latter clearly contradicts the assumption  $\lwprevious \phi \not\in r(0)$.
			
			\item Suppose $n>0$ and $(\system, r, n-1) \models  \phi$. We have by the induction hypothesis that  $\phi \in r(n-1)$. Thus, $\phi \in \Gamma_{n-1}$, and hence, by Lemma \ref{lem:R-next}, $\lsprevious \phi \in \Delta_n$. By axiom $\swprevax$, we get   $\lwprevious \phi \in \Delta_n = r(n)$, which is a contradiction.
		\end{itemize}
		
		$(\Leftarrow)$ Suppose $\lwprevious \phi \in r(n)$ and  $n>0$.  Since $\RO{r(n-1)}{r(n)}{\overline{\Gamma_{n-1}}}{\overline{\Delta_{n}}}$, we have $\lwprevious \phi \in \Delta_{n}$, by Lemma \ref{lem:R-next}, we get $\phi \in \Gamma_{n-1}$. By the induction hypothesis, $(\system, r, n-1) \models  \phi$, and hence $(\system, r, n) \models \lwprevious \phi$. 

		\item $\psi = \psi_1 \luntil \psi_2$.
		
		$(\Rightarrow)$ If $(\system, r, n) \models  \psi_1 \luntil \psi_2$, then $(\system, r, m) \models  \psi_2$ for some $m \geq n$, and $(\system, r, k) \models  \psi_1$ for all $k$ with $n \leq k < m$. By the induction hypothesis we get $\psi_2 \in r(m)$, and $\psi_1 \in r(k)$ for all~$k$ with $n \leq k < m$. We have to show $\psi_1 \luntil \psi_2 \in r(n)$, which follows by induction on $m$ as follows:
		\begin{itemize}
			\item
			Base case $m=n$. Since $\psi_2 \in r(n) = r(m)$ and $\vdash_\CS \psi_2 \rightarrow (\psi_1 \luntil \psi_2)$, we obtain $\psi_1 \luntil \psi_2 \in r(n)$. 
			
			\item
			Suppose $m > n$.  It follows from the induction hypothesis that $\psi_1 \luntil \psi_2 \in r(n+1)$. Since $\RO{r(n)}{r(n+1)}{ \overline{\Gamma_n}}{ \overline{\Delta_{n+1}}}$, and hence $\psi_1 \luntil \psi_2 \in \Delta_{n+1}$. Thus, by Lemma \ref{lem:R-next}, $\lnext (\psi_1 \luntil \psi_2) \in \Gamma_n$. Now assume to obtain a contradiction that $\psi_1 \luntil \psi_2 \not\in \Gamma_n$. Hence $\neg(\psi_1 \luntil \psi_2) \in \Gamma_n$. By axiom $(\luntil 2)$, 
			\[\vdash_\CS  \neg (\psi_1 \luntil \psi_2) \rightarrow [\neg \psi_2 \wedge (\neg \psi_1 \vee \neg \lnext (\psi_1 \luntil \psi_2))], \]
			and thus 
			\[\vdash_\CS  \neg (\psi_1 \luntil \psi_2) \wedge \psi_1 \rightarrow \neg \lnext (\psi_1 \luntil \psi_2), \]
			Thus, $\neg \lnext (\psi_1 \luntil \psi_2) \in \Gamma_n$, which is a contradiction. Thus, $\psi_1 \luntil \psi_2 \in \Gamma_n$ and hence $\psi_1 \luntil \psi_2  \in r(n)$.
			
		\end{itemize}
		
		$(\Leftarrow)$ If $\psi_1 \luntil \psi_2 \in r(n)$, then since $(r(0), r(1), \ldots, r(n),r(n+1),\ldots)$ is an acceptable sequence there exists $m \geq n$ such that $\psi_2 \in r(m)$, and $\psi_1 \in r(k)$ for all~$k$ with $n \leq k < m$. By the induction hypothesis we obtain 
		$(\system, r, m) \models  \psi_2$, and $(\system, r, k) \models  \psi_1$ for all $k$ with $n \leq k < m$. Thus 
		$(\system, r, n) \models  \psi_1 \luntil \psi_2$.
		
		\item $\psi = \psi_1 \lsince \psi_2$. 
		
		$(\Rightarrow)$ If $(\system, r, n) \models \psi_1 \lsince \psi_2$, then $(\system, r, m) \models \psi_2$ for some $m \leq n$, and $(\system, r, k) \models \psi_1$ for all $k$ with $m < k \leq n$. By the induction hypothesis, $\psi_2 \in r(m)$, and $\psi_1 \in r(k)$ for all $k$ with $m < k \leq n$. We want to show that $\psi_1 \lsince \psi_2 \in r(n)$. We prove it by induction on $m$ as follows.
		\begin{itemize}
			\item 		Base case $m=n$. Since $\psi_2 \in r(n)=r(m)$ and $\vdash_\CS  \psi_2 \rightarrow (\psi_1 \lsince \psi_2)$, we obtain $\psi_1 \lsince \psi_2 \in r(n)$. 
			
			\item 		Suppose $m < n$. Since $\RO{r(n-1)}{r(n)}{\overline{\Gamma_{n-1}}}{\overline{\Delta_{n}}}$, it follows from the induction hypothesis that $\psi_1 \lsince \psi_2 \in r(n-1)$, and hence $\psi_1 \lsince \psi_2 \in \Gamma_{n-1}$. Thus, by Lemma \ref{lem:R-next}, $\lsprevious (\psi_1 \lsince \psi_2) \in \Delta_n$. Now assume to obtain a contradiction that $\psi_1 \lsince \psi_2 \not\in \Delta_n$. Hence $\neg(\psi_1 \lsince \psi_2) \in \Delta_n$. By axiom  $\stwoax$,
			\[\vdash_\CS  \neg (\psi_1 \lsince \psi_2) \rightarrow [\neg \psi_2 \wedge (\neg \psi_1 \vee \neg \lsprevious (\psi_1 \lsince \psi_2))], \]
			and thus 
			\[\vdash_\CS  \neg (\psi_1 \lsince \psi_2) \wedge \psi_1 \rightarrow \neg \lsprevious (\psi_1 \lsince \psi_2), \]
			Thus, $\neg \lsprevious (\psi_1 \lsince \psi_2) \in \Delta_n$, which is a contradiction.
			
			$(\Leftarrow)$ Suppose $\psi_1 \lsince \psi_2 \in r(n)$. By Lemma \ref{lem: MCS temporal properties}, there is $m \leq n$ such that $\psi_2 \in r(m)$, and $\psi_1 \in r(k)$ for all $k$ with $m < k \leq n$. By the induction hypothesis, $(\system, r, m) \models \psi_2$ and $(\system, r, k) \models \psi_1$ for all $k$ with $m < k \leq n$, and thus $(\system, r, n) \models \psi_1 \lsince \psi_2$ as desired. 
		\end{itemize}

		\item $\psi = \jbox{t}_\agent \phi$.

		$(\Rightarrow)$ If $(\system, r, n) \models \jbox{t}_\agent \phi$, then $\phi \in \evidence_\agent (r(n),t)$. Thus, by the definition of $\evidence_\agent$, $\jbox{t}_\agent \phi \in \Gamma$, where $r(n) = \overline{\Gamma}$, and hence $\jbox{t}_\agent \phi \in r(n)= \overline{\Gamma_n}$. 
		
		$(\Leftarrow)$ If $\jbox{t}_\agent \phi \in r(n)$, then $\jbox{t}_\agent \phi \in \bigcap_{r(n) = \overline{\Gamma}} \Gamma$. Thus, $\phi \in \evidence_\agent (r(n),t)$. 
		Now suppose that $r(n) R_\agent  r'(n')$ and let $r'(n') = \overline{\Delta}$, for some $\Delta \in \MCS$. By the definition of $R_\agent$, there is $\Gamma \in \MCS$ such that $r(n) = \overline{\Gamma}$ and  
		$\{ \phi \ |\  \jbox{t}_\agent \phi \in \Gamma   \text{ for some $t$}\} \subseteq \Delta$. We find $\phi \in \Delta$. By the induction hypothesis, we get $(\system, r', n') \models \phi$. Since $r'$ and $n'$ were arbitrary, we conclude that $(\system, r, n) \models \jbox{t}_\agent \phi$. 

		\item $\psi = \jboxOAgent{t} \phi$.
		
		The proof for this case is similar to that of the previous case. \qed
	\end{itemize}
\end{proof}

\begin{theorem}[Completeness]\label{thm:Completeness-interpreted systems}
	Let $\CS$ be a constant specification for $\JTO$. For each formula $\chi$,
	\[
	\models_{\CS}  \chi  \quad\text{if{f}}\quad  \vdash_\CS \chi.
	\]
\end{theorem}
\begin{proof}
	Soundness is straightforward. For completeness, suppose that $\not \vdash_\CS \chi$. Thus, $\{ \neg \chi\}$ is a consistent set.
	Therefore, there exists $\Gamma \in \MCS$ with $\neg \chi \in \Gamma$. Let $\overline{\Gamma} = \Gamma \cap  \Subf^+(\chi)$.
	By Theorem \ref{thm: acceptable sequence}, there exists an acceptable sequence containing $\overline{\Gamma}$, say $(X_{0},X_1,\ldots)$, and there exists $n \geq 0$ such that $\overline{\Gamma} = X_n$. Define the run $r$ as follows $r(i) \colonequals X_i$. The run $r$ is in the system $\runs$ of the $\chi$-canonical F-interpreted system $\system$ for $\JTO_\CS$. 
	Since $\chi \not \in r(n)$, by the Truth Lemma, $(\system, r, n) \not\models \chi$. Therefore, $\not \models_{\CS} \chi$.  
\end{proof}

\begin{theorem}[Completeness]\label{thm:Weak Completeness-interpreted systems}
	Let $\CS$ be a constant specification for $\JTO$. For each formula $\chi$ and finite set of formulas $T$,
	\[
	T \models_{\CS}  \chi  \quad\text{if{f}}\quad  T \vdash_\CS \chi.
	\]
\end{theorem}
\begin{proof}
	Suppose that $T \not \vdash_\CS \chi$. Thus, $\not \vdash_\CS \bigwedge T \to \chi$. By Theorem \ref{thm:Completeness-interpreted systems}, there is an F-interpreted system $\system = (S, \runs, \ldots)$, $r \in \runs$, and $n \in \N$ such that $(\system, r, n) \models \bigwedge T$ and $(\system, r, n) \not\models \chi$. Therefore, $T \not \models_{\CS} \chi$.  
\end{proof}

\section{Interpreted-neighborhood systems}
\label{sec: Interpreted-neighborhood systems}

In normal modal logics, the dual of the necessity modality $\Box$ can be defined as a possibility modality $\Diamond$, by $\Diamond \phi := \neg \Box \neg \phi$. In some interpretations of modalities, the possibility modality plays an important role. For example, in the deontic logic framework, the dual of obligation is interpreted as permission ($\lpermissible \phi \equiv \neg \lobligatory \neg \phi$). One of the shortcomings of F-interpreted systems is that it does not provide a proper truth condition for the dual of (necessity) justification modalities. In fact, the truth condition for the dual of the justification modalities can be expressed as follows:
\begin{align*}
	(\system, r, n) &\entails \jdiamondAgent{t} \phi \text{ iff } \neg \phi \not \in \evidence_\agent(r(n),t)  \text { or }  \accrel_\agent(r(n)) \not\subseteq \truthset{\neg \phi}{\system}\, ,\\   
	(\system, r, n) &\entails \jboxPAgent{t} \phi \text{ iff } \neg \phi \not \in \evidenceO_\agent(r(n),t)  \text { or }  \accrelO_\agent(r(n)) \not\subseteq \truthset{\neg \phi}{\system}.
\end{align*}
In this section, we try to remedy this defect by introducing interpreted systems based on \Neigh models as an alternative semantics for $\JTO$. It is worth noting that Standefer \cite{Standefe-IGPL-2019} presented a neighborhood semantics for some justification logics. However, the base logics that are used in \cite{Standefe-IGPL-2019} are relevant logics rather than classical logic.

%


%

\begin{definition}\label{def: quasi interpreted sysytems}
	A \emph{quasi-interpreted-\Neigh system for $\JTO_\CS$} is a tuple 
	\[
	\system = (S, \runs, \neighborhood, \neighborhoodO,  \valuation)_{\agent \in \Ag}
	\] 
	where
	\begin{enumerate}
		
		\item $S$ is a non-empty set of states;
		
		\item 
		$\runs$ is a system on $S$;
		
		\item 
		$\neighborhood: \Im(\runs) \times \Terms^\lepistemic \to \powerset(\powerset(S))$ is an epistemic neighborhood function, for each $\agent \in \Ag$;
		
		\item 
		$\neighborhoodO: \Im(\runs) \times \Terms^\lobligatory \to \powerset(\powerset(S))$ is a normative neighborhood function, for each $\agent \in \Ag$;
		
%
		
		\item 
		$\valuation$ is a valuation function such that $\valuation: \Im(\runs) \to \powerset(\Prop)$ and $\valuation: S \setminus \Im(\runs) \to \powerset(\Formulae)$.  
		
		
	\end{enumerate}
\end{definition}

Given $X \subseteq S$,  $X \in \neighborhood(w,t)$ is intuitively read ``$t$ is a reason why proposition $X$ is known for agent $i$ in state $w$,'' and $X \in \neighborhoodO(w,t)$ is read ``$t$ is a reason why proposition $X$ is obligatory for agent $i$ in state $w$.'' 

Note that the states in $S \setminus \Im(\runs)$ may be considered as non-normal states, since there may be a state $w \in S \setminus \Im(\runs)$ and a formula $\phi$ such that both $\phi \in \valuation(w)$ and $\neg \phi \in \valuation(w)$. The following definition stipulates that states in $\Im(\runs)$ are normal.

\begin{definition}\label{def: truth interpreted systems}
	Let  $\system = (S, \runs, \neighborhood, \neighborhoodO,  \valuation)_{\agent \in \Ag}$ be a quasi-interpreted-\Neigh system. For $w \in S \setminus \Im(\runs)$, let $(\system, w) \entails \phi$ if $\phi \in \valuation(w)$. Truth at states $r(n) \in \Im(\runs)$ are inductively defined as follows: 
	\begin{align*}
	(\system, r(n)) &\entails p \text{ iff } p \in \valuation(r(n)) \, ,\\
	(\system, r(n)) &\not\entails \lfalse \, ,\\
	(\system, r(n)) &\entails \phi \limplies \psi \text{ iff } (\system, r(n)) \not\entails \phi \text{ or } (\system, r(n)) \entails \psi \, ,\\
	(\system, r(n)) &\entails \lwprevious \phi \text{ iff $n=0$ or } (\system, r(n-1)) \entails \phi \, ,\\
	(\system, r(n)) &\entails \lnext \phi \text{ iff } (\system, r(n+1)) \entails \phi \, ,\\
	(\system, r(n)) &\entails \phi \lsince \psi \text{ iff there is some } m \leq n \text{ such that } (\system, r(m)) \entails \psi \\ & \qquad\qquad \text{ and } (\system, r, k) \entails \phi \text{ for all $k$ with } m < k \leq n \, ,\\
	(\system, r(n)) &\entails \phi \luntil \psi \text{ iff there is some } m \geq n \text{ such that } (\system, r(m)) \entails \psi \\ & \qquad\qquad \text{ and } (\system, r(k)) \entails \phi \text{ for all $k$ with  } n \leq k < m \, ,\\   
	(\system, r(n)) &\entails \jboxAgent{t} \phi \text{ iff }  \truthsetModel{\phi} \in \neighborhood(r(n),t),\\
	(\system, r(n)) &\entails \jboxOAgent{t} \phi \text{ iff }  \truthsetModel{\phi} \in \neighborhoodO(r(n),t),
	\end{align*}
where $\truthsetModel{\phi} = \{ w \in S \mid  (\system, w) \entails \phi \}$ is the \textit{truth set} of $\phi$ (or the \textit{proposition} expressed by $\phi$).
\end{definition}


\begin{definition}\label{def: interpreted sysytems}
	An \emph{interpreted-\Neigh system}  
	\(
	\system = (S, \runs, \neighborhood, \neighborhoodO,  \valuation)_{\agent \in \Ag}
	\) 
	 for $\JTO_\CS$ is a quasi-interpreted-\Neigh system such that the neighborhood functions $\neighborhood$ and $\neighborhoodO$ satisfy the following conditions.\\
	 \textbf{Conditions on $\neighborhood$:}\\ For all terms $s,t \in \Terms^\lepistemic$, all formulas $\phi,\psi \in \Formulae$, all  $r(n) \in \Im(\runs)$, and all $i \in \Ag$:
	 \begin{enumerate}
	 	\setlength\itemsep{0.1cm}
	 	\item 
	 	$\truthsetModel{\phi} \in \neighborhood(r(n),t)$, for $\jbox{t}_\agent \phi \in \CS^\lepistemic$; \hfill \csNeighborhood
	 	
	 	\item 
	 	if $\truthsetModel{\phi \limplies \psi} \in \neighborhood(r(n),t)$ and $\truthsetModel{\phi} \in \neighborhood(r(n),s)$, then $\truthsetModel{\psi} \in \neighborhood(r(n), t \cdot s)$; \\ \text{}  \hfill \appNeighborhood
	 	
	 	\item 
	 	if $\truthsetModel{\phi} \in \neighborhood(r(n),s)$, then $\truthsetModel{\phi} \in \neighborhood(r(n), t + s) \cap\neighborhood(r(n), s + t) $; \hfill \sumNeighborhood
	 	
	 	\item 
	 	if $\truthsetModel{\phi} \in \neighborhood(r(n),t)$, then $r(n) \in \truthsetModel{\phi}$; \hfill \refNeighborhood
	 	
	 	\item 
	 	if $\truthsetModel{\phi} \in \neighborhood(r(n),t)$, then $\truthsetModel{\jboxAgent{t} \phi} \in \neighborhood(r(n),\tinspect t)$. \hfill \posintNeighborhood
	 \end{enumerate}
 \textbf{Conditions on $\neighborhoodO$:}\\ For all terms $s,t \in \Terms^\lobligatory$, all formulas $\phi,\psi \in \Formulae$, all  $r(n) \in \Im(\runs)$, and all $i \in \Ag$:
  \begin{enumerate}
  	\setlength\itemsep{0.1cm}
  	\item 
  	$\truthsetModel{\phi} \in \neighborhoodO(r(n),t)$, for $\jboxOAgent{t} \phi \in \CS^\lobligatory$; \hfill \csNeighborhoodO
  	
  	\item 
  	if $\truthsetModel{\phi \limplies \psi} \in \neighborhoodO(r(n),t)$ and $\truthsetModel{\phi} \in \neighborhoodO(r(n),s)$, then $\truthsetModel{\psi} \in \neighborhoodO(r(n), t \cdot s)$; \hfill \appNeighborhoodO
  	
  	\item 
  	if $\truthsetModel{\phi} \in \neighborhoodO(r(n),t)$, then $\truthsetModel{\neg \phi} \not \in \neighborhoodO(r(n),t)$; \hfill \nocNeighborhoodO

  	\item 
  	$\truthsetModel{\jboxOAgent{t} \phi \limplies \phi} \in \neighborhoodO(r(n),\tref t)$. \hfill \shiftrefNeighborhoodO
  	
  \end{enumerate}
  
\end{definition}

From the above definitions it follows that:  
\begin{align*}
(\system, r(n)) &\entails \jdiamondAgent{t} \phi \text{ iff }  \truthsetModel{\neg \phi} \not \in \neighborhood(r(n),t)\, ,\\
(\system, r(n)) &\entails \jboxPAgent{t} \phi \text{ iff }   \truthsetModel{\neg \phi} \not \in \neighborhoodO(r(n),t) \,.
\end{align*}

%


\begin{definition}
	Let $\CS$ be a constant specification for $\JTO$.
	\begin{enumerate}
		\item Given an interpreted-\Neigh system $\system = (S, \runs, \neighborhood, \neighborhoodO,  \valuation)_{\agent \in \Ag}$ for $\JTO_\CS$, we write $\system \entails \phi$ if
		for all $r(n) \in \Im(\runs)$, we have 
		$(\system, r(n)) \entails \phi$.
		
		\item We write $\entails^N_{\CS} \phi$ if $\system \entails \phi$ for all 
		interpreted-\Neigh systems $\system$ for $\JTO_\CS$.
		
		\item Given a set of formulas $T$ and a formula $\phi$ of $\JTO_\CS$, the (local) consequence relation is defined as follows: $T \models^N_\CS \phi$ iff for all 
		interpreted-\Neigh systems $\system = (S, \runs, \neighborhood, \neighborhoodO,  \valuation)_{\agent \in \Ag}$ for $\mathsf{L}_\CS$, for all  $r(n) \in \Im(\runs)$, if $(\system, r(n)) \entails \psi$ for all $\psi \in T$, then $(\system, r(n)) \entails \phi$.
	\end{enumerate}
\end{definition}

In the neighborhood semantics of modal logic we have the following regularity rule:
\[
\lrule{\entails \phi \liff \psi}{\entails \Box \phi \liff \Box \psi}\ (\sf{RE})
\]
However, the fact that $\entails^N_{\CS} \phi \liff \psi$ and the fact that $\phi$ is known (or is obligatory) for reason $t$ need not imply that $\phi$ is known (or is obligatory) for the same reason $t$. In other words, epistemic and normative reasons are hyperintentional (see \cite{Faroldi-Protopopescu-IGPL-2019} for a more detailed exposition). We show that a justification version of the rule $(\sf{RE})$ is not admissible in $\JTO$.

\begin{lemma}
	The regularity rule
	\[
	\lrule{\phi \liff \psi}{\jboxAgent{t} \phi \liff \jboxAgent{t} \psi}\ (\sf{JRE})
	\]
	is not validity preserving in $\JTO_\CS$, i.e. there exists $\phi, \psi \in \Formulae$ and $t \in \Terms^\lepistemic$ such that $\entails^N_{\CS} \phi \liff \psi$ but $\not \entails^N_{\CS} \jboxAgent{t} \phi \liff \jboxAgent{t} \psi$.
\end{lemma}

\begin{proof}
	Let $\CS$ be an arbitrary constant specification for $\JTO$. Consider the interpreted-\Neigh system $\system = (S, \runs, \neighborhood, \neighborhoodO,  \valuation)_{\agent \in \Ag}$ as follows:
	\begin{itemize}
		\setlength\itemsep{0.1cm}
		\item 
		$S = \{ w, v\}$,
		
		\item 
		$\runs = \{ r  \}$, where $r : \N \to S$ is defined by $r(n) = w$ for all $n \in \N$,
		
		\item 
		$\neighborhood(w,x) = \neighborhood(w,c) = \{ \{ w\} \} $, for all $x \in \VTerms$, for all $c \in \CTerms$, and all $i \in \Ag$,\\
		$\neighborhood(w,t) = \{ \{ w\}, S \}$, for all $t \in \Terms \setminus (\VTerms \cup \CTerms)$, and all $i \in \Ag$, 
				
		\item 
		$\neighborhoodO(w,t) = \{ \{ w\}, S \}$, for all $t \in \Terms$, and all $i \in \Ag$,

		\item 
		$\valuation(w) = \{ p \}$, $\valuation(v) = \{ p \wedge p \}$.
	\end{itemize}
	It is easy to show that $\system$ is an interpreted-\Neigh system for $\JTO_\CS$. Let $x \in \VTerms$. We observe that $\system \not \entails \jboxAgent{x} p \liff \jboxAgent{x} (p \wedge p)$ and thus $\not \entails^N_{\CS} \jboxAgent{x} p \liff \jboxAgent{x} (p \wedge p)$, but $\entails^N_{\CS} p \liff p \wedge p$.  \qed
\end{proof}

It is noteworthy that a mono-agent version of $(\sf{JRE})$ is used in the axiomatic formulation of some of the relevant justification logics of \cite{Standefe-IGPL-2019}. One might argue that for an equivalence to be valid in a relevant logic there must be some connection between the two equivalent propositions, and thus $(\sf{JRE})$ is expected to be admissible in this setting. Nonetheless, it is not still clear why two equivalent propositions should be known for the same reason. It seems that an argument in favor of $(\sf{JRE})$ in a relevant logic depends on a sensible notion of `connection', and thus this issue has yet to be investigated more.

\begin{lemma}
	The regularity rule
	\[
	\lrule{\phi \liff \psi}{\jboxOAgent{t} \phi \liff \jboxOAgent{t} \psi}\ (\sf{JRE}_\lobligatory)
	\]
	is not validity preserving in $\JTO_\CS$, i.e. there exists $\phi, \psi \in \Formulae$ and $t \in \Terms^\lobligatory$ such that $\entails^N_{\CS} \phi \liff \psi$ but $\not \entails^N_{\CS} \jboxOAgent{t} \phi \liff \jboxOAgent{t} \psi$.
\end{lemma}
\begin{proof}
	Let $\CS$ be an arbitrary constant specification for $\JTO$. Consider the interpreted-\Neigh system $\system = (S, \runs, \neighborhood, \neighborhoodO,  \valuation)_{\agent \in \Ag}$ as follows:
	\begin{itemize}
		\setlength\itemsep{0.1cm}
		\item 
		$S = \{ w, v\}$,
		
		\item 
		$\runs = \{ r  \}$, where $r : \N \to S$ is defined by $r(n) = w$ for all $n \in \N$,
		
		\item 
		$\neighborhood(w,t) = \{ \{ w \}, S \}$, for all $t \in \Terms$ and all $i \in \Ag$, 
%
		\item 
		$\neighborhoodO(w,x) = \neighborhoodO(w,c) = \{ \{ w\} \} $, for all $x \in \VTerms$, for all $c \in \CTerms$, and all $i \in \Ag$,\\
		$\neighborhoodO(w,t) = \{ \{ w\}, S \}$, for all $t \in \Terms \setminus (\VTerms \cup \CTerms)$, and all $i \in \Ag$, 
		
		\item 
		$\valuation(w) = \valuation(v) = \{ p \}$. 
	\end{itemize}
	%
	It is easy to show that $\system$ is an interpreted-\Neigh system for $\JTO_\CS$. Let $x \in \VTerms$. We observe that $\system \not \entails \jboxOAgent{x} p \liff \jboxOAgent{x} (p \wedge p)$ and thus $\not \entails^N_{\CS} \jboxOAgent{x} p \liff \jboxOAgent{x} (p \wedge p)$, but $\entails^N_{\CS} p \liff p \wedge p$.  \qed
\end{proof}

Next we show that $\conax$ is not valid in $\JTO$.  

\begin{lemma}
	Given a constant specification $\CS$ for $\JTO$  such that $\CS^\lobligatory = \emptyset$, we have:
	$$\not \entails^N_{\CS} \neg \jboxOAgent{t} \bot,$$
	for every $t \in \Terms^\lobligatory$.
\end{lemma}
\begin{proof}

	Let $\CS$ be a constant specification such that $\CS^\lobligatory = \emptyset$. Consider the interpreted-\Neigh system $\system = (S, \runs, \neighborhood, \neighborhoodO,  \valuation)_{\agent \in \Ag}$ as follows:
	\begin{itemize}
		\setlength\itemsep{0.1cm}
		\item 
		$S = \{ w, v\}$,
		
		\item 
		$\runs = \{ r  \}$, where $r : \N \to S$ is defined by $r(n) = w$ for all $n \in \N$,
		
		\item 
		$\neighborhood(w,t) = \{ \{ w \}, S \}$, for all $t \in \Terms$ and all $i \in \Ag$, 
		
		\item 
		$\neighborhoodO(w,t) =  \lset{ \lset{v}, S }$, for all $t \in \Terms$  and all $i \in \Ag$,

		\item 
		$\valuation(w) = \emptyset$, $\valuation(v) = \{ \bot \}  \cup \{ \jboxOAgent{t} \phi \limplies \phi \mid \phi \in \Formulae \}$.
	\end{itemize}
	It is easy to show that $\system$ is an interpreted-\Neigh system for $\JTO_\CS$. We observe that $(\system, r(n)) \not \entails \neg \jboxOAgent{t} \bot$.
 \qed
\end{proof}

Finally, we show that $\nocstrongax$ is not valid in $\JTO$.  

\begin{lemma}
	Given a constant specification $\CS$ for $\JTO$ such that $\CS^\lobligatory = \emptyset$, we have:
	$$
	\not \entails_{\CS} \jboxOAgent{s} \phi \limplies \jboxPAgent{t} \phi,
	$$
	 for some $\phi \in \Formulae$ and for some $s, t \in \Terms^\lobligatory$.
\end{lemma}
\begin{proof}

	Let $\CS$ be a constant specification such that $\CS^\lobligatory = \emptyset$. Consider the interpreted-\Neigh system $\system = (S, \runs, \neighborhood, \neighborhoodO,  \valuation)_{\agent \in \Ag}$ as follows:
	\begin{itemize}
		\setlength\itemsep{0.1cm}
		\item 
		$S = \{ w, v, u\}$,
		
		\item 
		$\runs = \{ r  \}$, where $r : \N \to S$ is defined by $r(n) = w$ for all $n \in \N$,
		
		\item 
		$\neighborhood(w,t) = \{ \{ w \}, S \}$, for all $t \in \Terms$ and all $i \in \Ag$, 
		
		\item 
		$\neighborhoodO(w,x) =  \lset{ \lset{ w, v } }$, for all $t \in \Terms$  and all $i \in \Ag$, \\
		$\neighborhoodO(w,y) =  \lset{ \lset{v, u} }$, for all $t \in \Terms$  and all $i \in \Ag$, \\
		$\neighborhoodO(w,c) =  \neighborhoodO(w,z) = \emptyset$, for all $z \in \VTerms$ s.t. $z \neq x,y$ and all $i \in \Ag$, \\
		$\neighborhoodO(w,t) =  \lset{ \lset{v}, \lset{ w, v } }$, for all $t \in \Terms$  and all $i \in \Ag$,

		\item 
		$\valuation(w) = \lset{ p }$, $\valuation(v) =\lset{ p, \neg p}  \cup \{ \jboxOAgent{t} \phi \limplies \phi \mid \phi \in \Formulae \}$, and $\valuation(u) = \lset{ \neg p }$.
	\end{itemize}
	It is easy to show that $\system$ is an interpreted-\Neigh system for $\JTO_\CS$. We observe that $(\system, r(n)) \entails  \jboxOAgent{x} p \wedge \jboxOAgent{y} \neg p$.
	\qed
\end{proof} 

In the rest of this section we show soundness and completeness of $\JTO_\CS$ with respect to interpreted-\Neigh systems. To this end, we first transform each F-interpreted system into an equivalent interpreted-\Neigh system, and then we use completeness of $\JTO_\CS$ with respect to F-interpreted systems.


\begin{lemma}\label{lem: Fitting to neighborhood un-conditional}
	Given an F-interpreted system $\system^F = (S, \runs,  \accrel_\agent, \evidence_\agent, \accrelO_\agent, \evidenceO_\agent,  \valuation)_{\agent \in \Ag}$ for $\JTO_\CS$,	there is an interpreted-\Neigh system $\system$ for $\JTO_\CS$ with the same set of runs which is equivalent to $\system^F$ with respect to $\Im(\runs)$, namely for all $\phi \in \Formulae$ and for all $r(n) \in \Im(\runs)$:
	\[
	(\system^F,r,n) \entails \phi  \mbox{ iff } (\system, r(n)) \entails \phi.
	\] 

\end{lemma}

\begin{proof}
	Suppose that $\system^F = (S, \runs,  \accrel_\agent, \evidence_\agent, \accrelO_\agent, \evidenceO_\agent,  \valuation)_{\agent \in \Ag}$ is an F-interpreted system for $\JTO_\CS$. Recall that the truth set of $\phi$ in $\system^F$ is defined as follows:
	$$\truthset{\phi}{{\system^F}} = \{ r(n) \in \Im(\runs) \mid  (\system^F, r,n) \entails \phi \}.$$
	 Let $\system = (S', \runs, \neighborhood, \neighborhoodO,  \valuation)_{\agent \in \Ag}$ be defined as follows:
	\begin{itemize}
		\setlength\itemsep{0.2cm}
		
		\item $S' := \Im(\runs)$.
			
		\item $\neighborhood (r(n),t) := \{ \truthset{\phi}{{\system^F}} \mid \phi \in \evidence_\agent(r(n),t) \text{ and } \accrel_\agent(r(n)) \subseteq \truthset{\phi}{{\system^F}} \}$, for all $t \in \Terms^\lepistemic$, all $r(n) \in \Im(\runs)$, and all $i \in \Ag$. 
		
		\item $\neighborhoodO (r(n),t) := \{ \truthset{\phi}{{\system^F}} \mid \phi \in \evidenceO_\agent(r(n),t) \text{ and } \accrelO_\agent(r(n)) \subseteq \truthset{\phi}{{\system^F}} \}$, for all $t \in \Terms^\lobligatory$, all $r(n) \in \Im(\runs)$, and all $i \in \Ag$.

		
		
	\end{itemize}
	%

	In order to prove that $\system$ and $\system^F$ are equivalent with respect to $\Im(\runs)$, we  show that for all $\phi \in \Formulae$:
	\begin{equation}\label{eq: Fitting to neighborhood un-conditional}
		\truthset{\phi}{{\system}} = \truthset{\phi}{{\system^F}}.
	\end{equation}
	
	This can be shown by induction on $\phi$, and we only show the case $\phi = \jboxAgent{t} \psi$ and $\phi = \jboxOAgent{t} \psi$. 
	\begin{itemize}
		\setlength\itemsep{0.1cm}
		\item $\phi = \jboxAgent{t} \psi$.
		
			$(\system^F, r, n) \entails \jboxAgent{t} \psi$ iff  $\psi \in \evidence_\agent(r(n),t)$ and $\accrel_\agent(r(n)) \subseteq \truthset{\psi}{{\system^F}}$ iff  $\truthset{\psi}{{\system^F}} \in \neighborhood(r(n),t)$, by the induction hypothesis, iff $\truthset{\psi}{{\system}} \in \neighborhood(r(n),t)$ iff $(\system, r(n)) \entails \jboxAgent{t} \psi$. 
			
		\item $\phi = \jboxOAgent{t} \psi$.
		
		$(\system^F, r, n) \entails \jboxOAgent{t} \psi$ iff  $\psi \in \evidenceO_\agent(r(n),t)$ and $\accrelO_\agent(r(n)) \subseteq \truthset{\psi}{{\system^F}}$ iff  $\truthset{\psi}{{\system^F}} \in \neighborhoodO(r(n),t)$, by the induction hypothesis, iff $\truthset{\psi}{{\system}} \in \neighborhoodO(r(n),t)$ iff $(\system, r(n)) \entails \jboxOAgent{t} \psi$. 
			
	\end{itemize}

	It is easy to show that $\system$ is an interpreted-\Neigh system for $\JTO_\CS$. As an illustration, we show that $\system$ satisfies $\appNeighborhood$.
	
	Suppose $\truthset{\phi \limplies \psi}{{\system}} \in \neighborhood(r(n),t)$ and $\truthset{\phi}{{\system}} \in \neighborhood(r(n),s)$. By \eqref{eq: Fitting to neighborhood un-conditional}, we have $\truthset{\phi \limplies \psi}{{\system^F}} \in \neighborhood(r(n),t)$ and $\truthset{\phi}{{\system^F}} \in \neighborhood(r(n),s)$. Thus, we get 
	$$
	\phi \limplies \psi \in \evidence(r(n),t), \quad \phi \in \evidence(r(n),s), \quad \accrel_\agent(r(n)) \subseteq \truthset{\phi}{{\system^F}}, \text{ and~~ } \accrel_\agent(r(n)) \subseteq \truthset{\phi \limplies \psi}{{\system^F}}.
	$$
	 By (application), we obtain  $\psi \in \evidence(r(n), t \tapp s)$. Moreover, since $\truthset{\phi}{{\system^F}} \cap \truthset{\phi \limplies \psi}{{\system^F}} \subseteq \truthset{\psi}{{\system^F}}$, we get $\accrel_\agent(r(n)) \subseteq \truthset{\psi}{{\system^F}}$. Therefore $\truthset{\psi}{{\system^F}} \in \neighborhood(r(n), t \tapp s)$, and $\truthset{\psi}{{\system}} \in \neighborhood(r(n), t \tapp s)$ as required. \qed
\end{proof}

\begin{theorem}[Completeness]\label{thm:Completeness-interpreted-neighborhood systems}
	Let  $\CS$ be a constant specification for $\JTO$. For each formula $\chi$,
	\[
	\models^N_{\CS}  \chi  \quad\text{if{f}}\quad  \vdash_\CS \chi.
	\]
\end{theorem}
\begin{proof}
	The soundness part is straightforward.	For completeness, suppose that $\not\vdash_\CS \chi$. By Theorem \ref{thm:Completeness-interpreted systems}, there is an F-interpreted system $\system^F = (S, \runs,  \accrel_\agent, \evidence_\agent, \accrelO_\agent, \evidenceO_\agent,  \valuation)_{\agent \in \Ag}$, and there is $r(n) \in \Im(\runs)$ such that $(\system^F, r, n) \not\models \chi$. By Lemma \ref{lem: Fitting to neighborhood un-conditional}, there is an interpreted-\Neigh system $\system$ for $\JTO_\CS$ with the same set of runs which is equivalent to $\system^F$ with respect to $\Im(\runs)$. Hence,  
	  $(\system, r(n)) \not\models \chi$. Therefore, $\not\models^N_{\CS} \chi$.  \qed
\end{proof}

\begin{theorem}[Completeness]\label{thm:Weak Completeness-interpreted-neighborhood systems}
	Let  $\CS$ be a constant specification for $\JTO$. For each formula $\chi$ and finite set of formulas $T$,
	\[
	T \models^N_{\CS}  \chi  \quad\text{if{f}}\quad  T \vdash_\CS \chi.
	\]
\end{theorem}
\begin{proof}
Follows from Theorem \ref{thm:Completeness-interpreted-neighborhood systems}.  \qed
\end{proof}

\section{The Protagoras Paradox}
\label{sec: Protagoras Paradox}

In this section we attempt to solve the Protagoras Paradox from a logical point of view. In particular, we focus on both the temporal aspects of the paradox and on those grounds and pieces of evidence, namely normative reasons, that agents base their reasoning on. Various versions of the paradox have been proposed in the literature.\footnote{This paper do not argue about historical issues, such as what actually happened. For such issues refer to \cite{Sobel1987} and \cite{Jankowski2015}.} We consider the following version.

The paradox is about the famous Greek law teacher Protagoras who took a poor but talented 
student Euathlus and agreed to teach him without a fee on condition that 
after the student completed his studies, he would pay Protagoras a certain fee on the day when he win his first court-case\footnote{See \cite{Sobel1987} for a discussion on ``first win a case" or ``win first case" condition.}  (not necessarily as a lawyer)\footnote{If it is stipulated in the contract that Euathlus himself should plead before jurors and win his court-case, then one trivial way for Euathlus to avoid paying the fee is that he would hire a lawyer to plead before jurors. Because, if he win with a lawyer, his victory would be non-paradoxical, and if he lose with a lawyer, he would not yet have won his first case. } 
 

Protagoras taught Euathlus according to the contract. However, after the lessons were finished, Euathlus  did not
take any law cases. Some time elapsed and Protagoras sued his student for the teaching fee, and in court each pleads his own case. The following arguments were presented to the judge in the court.

Protagoras: If I win this case, then Euathlus has to pay me by virtue of
your verdict. On the other hand, if I do not win the case, then he will won his first case, hence he has to pay me, this time by virtue of our agreement. 
In either case, he has to pay me. Therefore, he has to pay me my fee.

Euathlus: If Protagoras wins the case, then I will not yet have won my first case, so, by our agreement, I don't have to pay. If Protagoras does not win the case, then, by your verdict, I don't have to pay. In either case, I don't have to pay him. Hence I don't have to
pay the fee.\footnote{We suppose here that cases are numbered according to their commencement	dates, rather than their conclusion dates. In addition, the court resolves the case by a yea or nay verdict, and not by dismissal or suspension. Otherwise, we need a four-valued logic for formalization. In fact, in an ancient rendition of the case (e.g. Aulus Gellius, The Attic Nights, (f.c. 150 A.D.)), the court decides to postpone its decision (deferral), and in another version (e.g. Hermogenes)  the case is dismissed as not really fitting a courtroom. For more details see \cite{Jankowski2015}.}\\

For the history and other versions of the paradox see \cite{Sobel1987,Jankowski2015}, and for a formalization of the paradox see \cite{Aqvist1995,Lenzen1977,Lukowski2011,Glavanicova-Pascucci-DEON-2021}.

There are two aspects of the problem that have to be considered. One is  the logical
question of how to resolve the paradoxical situation that both
arguments seem cogent, yet have contradictory conclusions (this issue is addressed in Section \ref{sec: Formalization of the paradox}. The other is the legal question
of what the court should rule (this issue is addressed in Section \ref{sec:Judges reasoning}). 

One could recognize a self-reference argument in the paradox. For example, Sobel writes \cite[page 10]{Sobel1987}:

\begin{quote}
The court could realize that were it to rule that
there \textit{was} an obligation under the contract and rule against the
student in his first case, there would be \textit{no} obligation under the contract. And it could realize that were it to find that there was \textit{no} obligation under the contract and rule for the student in his first case, there would \textit{be} an obligation under the contract.
\end{quote}

Some authors conclude that this specific court case about the contract is not the first case as referred to in the contract, and thus the contract is not applicable in this court case. In contrast, we think that the court can apply the contract in this case, and the root of the contradiction is somewhere else. Indeed, as Jankowski mentioned \cite[page 103]{Jankowski2015}: 
\begin{quote}
	``the ‘perplexity’ can
	be avoided if a temporal element is taken into consideration."
\end{quote}

 In fact, we will see how formalizing the paradox using temporal operators vanishes the problem of self-reference. We also argue that the paradox can be satisfactorily resolved if normative reasons are mentioned explicitly in the arguments of Protagoras and Euathlus.

\subsection{Formalization in $\JTO$}
\label{sec: Formalization of the paradox}

In the rest of this paper, I aim to formalize the paradox in $\JTO_\CS$. Since the constant specification $\CS$ plays no role in our analysis, we assume that $\CS$ is an arbitrary constant specification throughout this section. Consider the following dictionary:\\

\noindent Agents:\\
$\protagoras:=$ Protagoras.\\
$\euathlus :=$ Euathlus.\\
$\judge :=$ Judge.\\

\noindent Atomic propositions:\\
$\winP :=$ Protagoras wins the case (or judge awards the case to Protagoras).\\
$\winE :=$ Euathlus wins his first court-case.\\
$\pay :=$ Euathlus pays Protagoras his fee. \\
$\sueP :=$ Protagoras sues Euathlus.\\

\noindent Justification variables:\\
$\agreement :=$ The agreement reached between Protagoras and Euathlus.\\
$\verdictP :=$ Judge's verdict in favor of Protagoras.\\
$\verdictE :=$ Judge's verdict in favor of Euathlus. \\

Thus, the set of agents are $\Ag = \lset{ \protagoras, \euathlus, \judge }$.
Moreover, let the time of agreement be 0, let the time that Protagoras sues  Euathlus be 5, and let the time of pronouncement of the court be 10.


\paragraph{\textbf{Formalization}.}

First let us formalize the proposition ``Euathlus \textit{has to} pay Protagoras." Note that Euathlus has to pay iff he has an obligation (of some kind, by terms of the contract or by judge's verdict) to pay. This can be formalized by
$$
\jboxO{t}{\euathlus} \pay,
$$
where $t$ is a normative reason why Euathlus has to pay the fee.

Next let us formalize the contract between Protagoras and  Euathlus. One possible simple formalization is as follows:
\[
\contract := \lalwaysPastFuture(\winE \liff  \jboxO{\agreement}{\euathlus} \pay).
\]

It is worth noting that in  \cite{Lenzen1977,Aqvist1995,Glavanicova-Pascucci-DEON-2021} the contract between Protagoras and Euathlus is formalized by means of a necessity modality. Thus, instead of the temporal necessity operator `always' $\lalwaysPastFuture$ a necessity operator $L$ is used, and the contract can be formalized as follows:
\[
L(\winE \liff  \jboxO{\agreement}{\euathlus} \pay),
\]
which is read as ``it is necessary to see to it that Euathlus has an obligation to pay by the terms of the agreement if and only if he wins his first case." Although it is possible to extend the language of $\JTO$ with a necessity modality $L$, I will not pursue the contribution of this notion to the logic $\JTO$, and so I will proceed with the temporal necessity operator $\lalwaysPastFuture$ as is used in $\contract$.

Then we need the following assumption about the court:
\begin{align*}
	\court  := &\ltrue_{10}[(\neg \winP \leftrightarrow \winE) \wedge\\
	&  (\winP \limplies \jboxO{\verdictP}{\euathlus} \pay ) \wedge \\
	& (\neg \winP \limplies \neg \jboxO{\verdictE}{\euathlus} \pay)]. 
\end{align*}
%
$\court$ says that ``at time 10 (the time of the court pronouncement) we have that Protagoras does not win the court case if and only if Euathlus wins his first case, and if judge awards the case to Protagoras then according to the judge's verdict ($\verdictP$) Euathlus has to pay the fee, and if judge awards the case to Euathlus then according to the judge's verdict ($\verdictE$) there is no obligation for Euathlus to pay." Note that these statements only hold for this special court case (which is at the same time Euathlus' first case), and may not hold for other cases that Protagoras and Euathlus are involved, and thus these statements only hold at time 10.

The set of assumptions used in the formalization of the paradox is:
\[
\Gamma := \{\contract, \court\}. 
\]
%




Let us now formalize the  reasoning  of Protagoras and Euathlus presented in the court. Note that both Protagoras' and Euathlus's arguments take place between time instances 5 and 10, but both arguments refer to state of affairs at $\ltime = 10$.  

\paragraph{\textbf{Protagoras reasoning}:}

If I win this case, then Euathlus has to pay me by virtue of
your verdict.

\begin{enumerate}
	\item $\ltime = 10$, \hfill hypothesis
		
	\item $\winP \limplies   \jboxO{\verdictP}{\euathlus} \pay$ \hfill by $\court$ and Lemma \ref{lem: temporal truth predicate properties}
\end{enumerate}
On the other hand, if I do not win the case, then he will won his first case, hence he has to pay me, this time by virtue of our agreement.

\begin{enumerate}[resume]
	
	\item $\neg \winP \limplies \winE$ \hfill by $\court$ and Lemma \ref{lem: temporal truth predicate properties}
	
	
	
	\item $\winE \limplies   \jboxO{\agreement}{\euathlus} \pay$ \hfill by $\contract$
		
	\item $\neg \winP \limplies  \jboxO{\agreement}{\euathlus} \pay$ \hfill from 3 and 4
\end{enumerate}
In either case, he has to pay me. 

\begin{enumerate}[resume]
	
	\item $\winP \vee \neg \winP$ \hfill by propositional reasoning
%
%
	
	\item $ \jboxO{\verdictP}{\euathlus} \pay \vee  \jboxO{\agreement}{\euathlus} \pay$ \hfill from 2, 5, and 6
\end{enumerate}
Therefore, he has to pay me my fee.  
\begin{equation}\label{eq:Protagoras reasoning}
\Gamma, \ltime = 10 \vdash_\CS  \jboxO{\verdictP}{\euathlus} \pay \vee  \jboxO{\agreement}{\euathlus} \pay.
\end{equation}

\paragraph{\textbf{Euathlus reasoning}:}

If Protagoras wins the case, then I will not yet have won my first case, so, by our agreement, I don't have to pay.

\begin{enumerate}
	\item $\ltime = 10$ \hfill hypothesis
	
	\item $\winP \limplies \neg \winE$ \hfill by $\court$ and Lemma \ref{lem: temporal truth predicate properties}
	

	\item $\neg \winE \limplies  \neg \jboxO{\agreement}{\euathlus} \pay $ \hfill by $\contract$
	
	\item $\winP \limplies  \neg \jboxO{\agreement}{\euathlus} \pay$ \hfill from 2 and 3
\end{enumerate}
%
%
%
%
%
%
%
If Protagoras does not win the case, then, by your verdict, I don't have to pay.

\begin{enumerate}[resume]
	
	\item $\neg \winP \limplies  \neg \jboxO{\verdictE}{\euathlus} \pay$ \hfill by $\court$  and Lemma \ref{lem: temporal truth predicate properties}
\end{enumerate}
In either case, I don't have to pay him.

\begin{enumerate}[resume]
	
	\item $\winP \vee \neg \winP$ \hfill by propositional reasoning
%
%
	
	\item $ \neg \jboxO{\agreement}{\euathlus} \pay \vee   \neg \jboxO{\verdictE}{\euathlus} \pay$ \hfill from  4, 5, and 6 
\end{enumerate}
Hence I don't have to pay the fee.
\begin{equation}\label{eq:Euathlus reasoning}
\Gamma, \ltime = 10 \vdash  \neg \jboxO{\agreement}{\euathlus} \pay \vee   \neg \jboxO{\verdictE}{\euathlus} \pay.
\end{equation}

It is obvious that \eqref{eq:Protagoras reasoning} and \eqref{eq:Euathlus reasoning} lead to no contradiction. In fact, using the temporal deontic justification logic $\JTO$, we are able to distinguish two kinds of obligations for Euathlus:  
\begin{itemize}
	\item 
	A conventional obligation (which can be considered as a kind of moral obligation here), expressed by the formula $\jboxO{\agreement}{\euathlus} \pay$.
	
	\item 
	A juridical obligation (which can be considered as a kind of legal obligation here), expressed by the formula $\jboxO{\verdictE}{\euathlus} \pay$.
\end{itemize}
In fact, the next lemma shows that it is not possible to derive a contradiction from the above set of assumptions.

\begin{lemma}\label{lem:the set of assumptions in the paradox is consistent}
	$\Gamma \cup \lset{ \ltime = 10 }$ is consistent in $\JTO_\CS$.
\end{lemma}
\begin{proof}
	Consider the interpreted-\Neigh system $\system_1 = (S, \runs, \neighborhood, \neighborhoodO,  \valuation)_{\agent \in \Ag}$ as follows: 
	\begin{itemize}
		\setlength\itemsep{0.1cm}
		\item 
		$S = \{ w_n \mid n \in \N\}$,
		
		\item 
		$\runs = \{ r  \}$, where $r : \N \to S$ is defined by $r(n) = w_n$ for all $n \in \N$,
		
		\item 
		$\neighborhood(r(n),t) = \{ S \}$, for all $t \in \Terms$, all $i \in \Ag$, and  all $n \in \N$, 
		%
		\item 
		$\neighborhoodO(r(10),\agreement) = \lset{ S }$, for all  $i \in \Ag$,  \\
		$\neighborhoodO(r(n),\agreement) = \emptyset$, for all  $i \in \Ag$, and  all $n \neq 10$ \\
		$\neighborhoodO(r(n),x) = \emptyset$, for all $x \in \VTerms$ such that $x \neq \agreement$, all $i \in \Ag$, and  all $n \in \N$ \\
		$\neighborhoodO(r(n),c) =  \neighborhoodO(r(n),t) = \lset{ S }$, for all $c \in \CTerms$, and all $t \in \Terms \setminus \VTerms$, all $i \in \Ag$, and  all $n \in \N$, 
		
		\item 
		$\valuation(r(n)) =  \lset{ \pay }$, for all $n \in \N$, and	$\valuation(r(10)) = \lset{ \winE }$.
	\end{itemize}
	%
	It is easy to show that $\system_1$ is an interpreted-\Neigh system for $\JTO_\CS$. 
	In addition, 
	$$
	(\system_1, r(10)) \models \contract \wedge \court \wedge \ltime = 10.
	$$
	 The
	result thus follows by the completeness theorem (Theorem \ref{thm:Weak Completeness-interpreted-neighborhood systems}).\qed
\end{proof}

\paragraph{\textbf{Formalization in modal logic}.}

	It is worth noting that if Protagoras' and Euathlus' reasonings are formalized within the framework of modal logic, then we get an explicit contradiction at the time of pronouncement. Thus, we consider here a combination of the standard deontic logic, $\Logic{SDL}$, and the linear temporal logic (with past time operators), $\Logic{LTL}$. The resulting logic is denoted by $\Logic{SDL \oplus LTL}$. In order to formalize the above inferences in $\Logic{SDL \oplus LTL}$ we use the \textit{forgetful projection} (cf. \cite{Art01BSL}). Let $\phi$ be a formula in the language of $\JTO$. The forgetful projection of $\phi$, denoted by $\phi^\circ$, is the result of replacing all occurrences of $\jboxOAgent{t} \psi$ with $\lobligatory_\agent \psi$ in the formula $\phi$. 
	
	Now consider the following set of assumptions 
	\[
	\Gamma^\circ = \{ \contract^\circ, \court^\circ \},
	\]
	where
		\[
	\contract^\circ := \lalwaysPastFuture (\winE \liff  \lobligatory_\euathlus \pay),
	\]
	and
	\begin{align*}
	\court^\circ &:= \ltrue_{10}[(\neg \winP \leftrightarrow \winE) \wedge  (\winP \limplies  \lobligatory_\euathlus \pay) \wedge (\neg \winP \limplies  \neg \lobligatory_\euathlus \pay)] \\
	&\equiv \ltrue_{10}[(\neg \winP \leftrightarrow \winE) \wedge  (\winP \liff  \lobligatory_\euathlus \pay)].
	\end{align*}
	An argument similar to the above arguments of Protagoras and Euathlus (see \eqref{eq:Protagoras reasoning} and \eqref{eq:Euathlus reasoning}) shows that:
	\begin{equation}\label{eq:Protagoras reasoning-standard deontic logic}
	\Gamma^\circ, \ltime = 10 \vdash_\Logic{SDL \oplus LTL}  \lobligatory_\euathlus \pay.
	\end{equation}
	\begin{equation}\label{eq:Euathlus reasoning-standard deontic logic}
	\Gamma^\circ, \ltime = 10 \vdash_\Logic{SDL \oplus LTL}  \neg \lobligatory_\euathlus \pay.
	\end{equation}
	In contrast to the arguments of Protagoras and Euathlus in $\JTO$ (see \eqref{eq:Protagoras reasoning} and \eqref{eq:Euathlus reasoning}), a contradiction can be derived from \eqref{eq:Protagoras reasoning-standard deontic logic} and \eqref{eq:Euathlus reasoning-standard deontic logic}:
	\begin{equation}\label{eq:paradox contradiction}
	\Gamma^\circ, \ltime = 10 \vdash_\Logic{SDL \oplus LTL} \bot.
	\end{equation}
	%
	This shows that any verdict passed by the court (at the time of pronouncement, i.e. $\ltime = 10$) reach a contradiction. Contrary to our observation, Goossens (\cite[page 72]{Goossens1977}) and Sobel (\cite[page 7]{Sobel1987}) try to show  that no explicit contradiction can be derived from the Protagoras/Euathlus arguments. 


\paragraph{\textbf{A more precise formalization}.}

Suppose that at time $n$ we have that Euathlus wins his first case (or judge decides against Protagoras in the aforementioned court case). In the previous formalization of the contract and the court (see the formulas $\contract$ and $\court$), it is stipulated that if Euathlus wins his first case at time $n$ (or if judge decides against Protagoras at time $n$), then Euathlus is obliged to pay the education fee at the same time, i.e. at time $n$. It seems more precise to say that the obligation of paying the fee holds in a time interval (from the time $n$ until Euathlus pays the fee, provided he pays the fee), and moreover, the obligation is started from the next time, say $n+1$. Thus, instead of stipulating that Euathlus is obliged to pay in the same time of his winning (or in the same time of judge's pronouncement), formalized by $\jboxO{t}{\euathlus} \pay$ for some normative reason $t$, it is more natural to stipulate that at the next time he is obliged to pay until he pays, formalized by the formula $\lnext \jboxO{t}{\euathlus} \pay \luntil \pay$. Further, it is implicit in the contract that from tomorrow of the day of contract until Euathlus wins his first case there is no obligation for him to pay. The latter can be formalized by $(\lnext \neg \jboxO{\agreement}{\euathlus} \pay) \luntil \winE$. Therefore, one possible formalization is as follows:
\begin{align}\label{eq: contract with Until}
	\lalwaysPastFuture [(\lnext \neg \jboxO{\agreement}{\euathlus} \pay \luntil \winE) \wedge (\winE \limplies  (\lnext \jboxO{\agreement}{\euathlus} \pay \luntil \pay))].
\end{align}
%
But then from this, using axiom $\uoneax$,  we get
\[
\lalwaysPastFuture [(\lnext \neg \jboxO{\agreement}{\euathlus} \pay \luntil \winE) \wedge (\winE \limplies (\lnext \jboxO{\agreement}{\euathlus} \pay \luntil \pay))]
 \vdash_{\CS} \leventually \winE,
\]
and 
\[
\lalwaysPastFuture [(\lnext \neg \jboxO{\agreement}{\euathlus} \pay \luntil \winE) \wedge (\winE \limplies (\lnext \jboxO{\agreement}{\euathlus} \pay \luntil \pay))]
\vdash_{\CS}  \winE \limplies \lnext \leventually \pay.
\]
Thus 
\[
\lalwaysPastFuture [(\lnext \neg \jboxO{\agreement}{\euathlus} \pay \luntil \winE) \wedge (\winE \limplies (\lnext \jboxO{\agreement}{\euathlus} \pay \luntil \pay))]
 \vdash_{\CS} \leventually  \jboxO{\agreement}{\euathlus} \pay,
\]
which is not compatible with the story, since nothing necessitates Euathlus to take a law case and thus nothing obligated him to pay his fee. In fact, the operator $\luntil$ is a \textit{strong} form of the until operator, because it demands that $\winE$ will hold sometime. In order to remedy this defect we employ the \textit{weak-until} or \textit{unless} operator, $\lunless$, instead of the until operator in the formalization of the contract. 
%
%
%
%
Hence, it seems a more realistic formalization would be as follows:
\[
\contract' := \lalwaysPastFuture [(\lnext \neg \jboxO{\agreement}{\euathlus} \pay \lunless \winE) \wedge (\winE \limplies (\lnext \jboxO{\agreement}{\euathlus} \pay \lunless \pay))].
\]
%
%
Let us recall that $\phi \lunless \psi := (\phi \luntil \psi) \vee (\lalways \phi)$. Thus, the first conjunct of $\contract'$ says that either from tomorrow on Euathlus has no obligation to pay the fee or  from tomorrow on he has no obligation to pay the fee until he wins his first case. The second conjunct says that if Euathlus wins his first case, then either from tomorrow on he has an obligation to pay the fee or from tomorrow on he has an obligation to pay the fee until he pays the fee.

Then we have

\begin{lemma}\label{lem: paradox 1}
	Let $\CS$ be an arbitrary constant specification for $\JTO$. Then
	\[
	\contract' \not \vdash_\CS \leventually \winE.
	\]
\end{lemma} 
\begin{proof}
	Let $\CS$ be an arbitrary constant specification. Consider the interpreted-\Neigh system $\system_2 = (S, \runs, \neighborhood, \neighborhoodO,  \valuation)_{\agent \in \Ag}$ as follows:
	\begin{itemize}
		\setlength\itemsep{0.1cm}
		\item 
		$S = \{ w\}$,
		
		\item 
		$\runs = \{ r  \}$, where $r : \N \to S$ is defined by $r(n) = w$ for all $n \in \N$,
		
		\item 
		$\neighborhood(w,t) = \{ S \}$, for all $t \in \Terms$ and all $i \in \Ag$, 
		%
		\item 
		$\neighborhoodO(w,t) = \{ S \}$, for all $t \in \Terms$, and all $i \in \Ag$, 
		
		\item 
		$\valuation(w) =  \emptyset$. 
	\end{itemize}
	%
	It is easy to show that $\system_2$ is an interpreted-\Neigh system for $\JTO_\CS$. 
	Since $(\system_2, r(n)) \models \lalways \lnext \neg \jboxO{\agreement}{\euathlus} \pay $, we get $(\system_2, r(n)) \models \lnext \neg \jboxO{\agreement}{\euathlus} \pay \lunless \winE$. From this and $(\system_2, r(n)) \models \neg \winE$,  it follows that $(\system_2, r(0)) \models \contract'$. In addition, $(\system_2, r(0)) \not\models \leventually \winE$. The
	result thus follows by the completeness theorem (Theorem \ref{thm:Weak Completeness-interpreted-neighborhood systems}). \qed
\end{proof}


Protagoras may have considered \eqref{eq: contract with Until} as the contract when he agreed to teach without any fee, and thought that Euathlus will win a case in the future. On the other hand, Euathlus considered $\contract'$ instead. According to the scenario (and $\contract'$), Euathlus could complete his course but he never wins any case (because he may take no law case in the future).

%
The above remark about having obligations in a time interval, instead of a certain moment in time, can be applied to the assumption about the court:
\begin{align*}
	\court'  := &\ltrue_{10}[(\neg \winP \leftrightarrow \winE) \wedge\\
	&  (\winP \limplies (\lnext \jboxO{\verdictP}{\euathlus} \pay \lunless \pay) ) \wedge \\
	& (\neg \winP \limplies \lalways \neg \jboxO{\verdictE}{\euathlus} \pay)].
\end{align*}
The second and third conjuncts in $\court'$ say that if judge awards the case to Protagoras then from tomorrow on  according to the judge's verdict ($\verdictP$) Euathlus has an obligation to pay the fee unless he pays the fee, and if judge awards the case to Euathlus then henceforth  there is no obligation for Euathlus to pay the fee in virtue of the judge's verdict ($\verdictE$).

Let $\Delta := \lset{ \contract', \court' }$ be the new set of assumptions. Then, similar to the above arguments of Protagoras and Euathlus (see \eqref{eq:Protagoras reasoning} and \eqref{eq:Euathlus reasoning}), one can show that: 
\begin{align}\label{eq: Protagoras and Euathlus arguments precise formalization}
	\Delta, \ltime = 10 &\vdash_\CS 
	(\lnext \jboxO{\verdictP}{\euathlus} \pay \lunless \pay)  \vee (\lnext \jboxO{\agreement}{\euathlus} \pay \lunless \pay).\\
	\Delta, \ltime = 10 &\vdash_\CS 
\lnext \neg \jboxO{\agreement}{\euathlus} \pay \vee  \lnext \neg \jboxO{\verdictE}{\euathlus} \pay.
\end{align}
The proofs are as follows:

\paragraph{Protagoras reasoning:}

\begin{enumerate}
	\item $\ltime = 10$ \hfill hypothesis
	
	\item $\winP \limplies  (\lnext \jboxO{\verdictP}{\euathlus} \pay \lunless \pay)$ \hfill by $\court'$ and Lemma \ref{lem: temporal truth predicate properties}
	
	\item $\neg \winP \limplies \winE$ \hfill by $\court'$ and Lemma \ref{lem: temporal truth predicate properties}
	
	
	
	\item $\winE \limplies (\lnext \jboxO{\agreement}{\euathlus} \pay \lunless \pay)$ \hfill by $\contract'$ and Lemma \ref{lem: basic results of linear temporal logic}
	
	\item $\neg \winP \limplies (\lnext \jboxO{\agreement}{\euathlus} \pay \lunless \pay)$ \hfill from 3 and 4
	
	\item $\winP \vee \neg \winP$ \hfill by propositional reasoning
	%
	%
	
	\item $(\lnext \jboxO{\verdictP}{\euathlus} \pay \lunless \pay)  \vee (\lnext \jboxO{\agreement}{\euathlus} \pay \lunless \pay)$ \hfill from 2, 5, and 6
\end{enumerate}

\paragraph{Euathlus reasoning:}

\begin{enumerate}
	\item $\ltime = 10$ \hfill hypothesis
	
	\item $\winP \limplies \neg \winE$ \hfill by $\court'$ and Lemma \ref{lem: temporal truth predicate properties}
	 
	
	\item $\lnext \neg \jboxO{\agreement}{\euathlus} \pay \lunless \winE$ \hfill by $\contract'$ and Lemma \ref{lem: basic results of linear temporal logic}
		
	\item $\winP \limplies (\lnext \neg \jboxO{\agreement}{\euathlus} \pay \lunless \winE)$ \hfill from 3 by propositional reasoning
	
	\item $\winP \limplies \neg \winE \wedge (\lnext \neg \jboxO{\agreement}{\euathlus} \pay \lunless \winE)$ \hfill from 2 and 4
	 
	\item $\winP \limplies \lnext \neg \jboxO{\agreement}{\euathlus} \pay \wedge \lnext (\lnext \neg \jboxO{\agreement}{\euathlus} \pay \lunless \winE)$ \hfill from 5 and Lemma \ref{lem: basic results of linear temporal logic}
	 
	\item $\winP \limplies \lnext \neg \jboxO{\agreement}{\euathlus} \pay$ \hfill from 6
	
	\item $\neg \winP \limplies  \lalways \neg \jboxO{\verdictE}{\euathlus} \pay$ \hfill by $\court'$  and Lemma \ref{lem: temporal truth predicate properties}

	\item $\neg \winP \limplies  \lnext \neg \jboxO{\verdictE}{\euathlus} \pay$ \hfill from 8 and Lemma \ref{lem: basic results of linear temporal logic}
	
	\item $\winP \vee \neg \winP$ \hfill by propositional reasoning
	%
	%
	
	\item $\lnext \neg \jboxO{\agreement}{\euathlus} \pay \vee  \lnext \neg \jboxO{\verdictE}{\euathlus} \pay$ \hfill from  7, 9, and 10 
\end{enumerate}

Then we show that no contradiction can be derived from the above set of assumptions.

\begin{lemma}\label{lem: second set of assumptions in the paradox is consistent}
	$\Delta \cup \lset{ \ltime = 10 }$ is consistent  in $\JTO_\CS$.
\end{lemma}
\begin{proof}
	Let $\CS$ be an arbitrary constant specification. Consider the interpreted-\Neigh system $\system_3 = (S, \runs, \neighborhood, \neighborhoodO,  \valuation)_{\agent \in \Ag}$ as follows: 
	\begin{itemize}
		\setlength\itemsep{0.1cm}
		\item 
		$S = \{ w_n \mid n \in \N\}$,
		
		\item 
		$\runs = \{ r  \}$, where $r : \N \to S$ is defined by $r(n) = w_n$ for all $n \in \N$,
		
		\item 
		$\neighborhood(r(n),t) = \{ S \}$, for all $t \in \Terms$, all $i \in \Ag$, and  all $n \in \N$, 
		%
		\item 
		$\neighborhoodO(r(n),\agreement) = \lset{ S }$, for all  $i \in \Ag$, and  all $n \in \N$, \\ 
		$\neighborhoodO(r(n),x) = \emptyset$, for all $x \in \VTerms$ such that $x \neq \agreement$, all $i \in \Ag$, and  all $n \in \N$, \\
		$\neighborhoodO(r(n),c) =  \neighborhoodO(r(n),t) = \lset{ S }$, for all $c \in \CTerms$, and all $t \in \Terms \setminus \VTerms$, all $i \in \Ag$, and  all $n \in \N$,

		\item 
		$\valuation(r(n)) =  \lset{ \winE, \pay }$, for all $n \in \N$.
	\end{itemize}
	%
	It is easy to show that $\system_3$ is an interpreted-\Neigh system for $\JTO_\CS$.
	For every $n \in \N$, we have that $(\system_3, r(n)) \models  \winE$, and hence $(\system_3, r(n)) \models \lnext \neg \jboxO{\agreement}{\euathlus} \pay \luntil \winE$. Thus, 
	$$
	(\system_3, r(n)) \models \lnext \neg \jboxO{\agreement}{\euathlus} \pay \lunless \winE,
	$$
	for every $n \in \N$. On the other hand, for every $n \in \N$, we have that $(\system_3, r(n)) \models \lalways \lnext  \jboxO{\agreement}{\euathlus} \pay$, and hence $(\system_3, r(n)) \models  \lnext  \jboxO{\agreement}{\euathlus} \pay \lunless \pay$. Thus, 
	$$
	(\system_3, r(n)) \models  \winE \limplies (\lnext  \jboxO{\agreement}{\euathlus} \pay \lunless \pay),
	$$
	for every $n \in \N$. Therefore, $(\system_3, r(10)) \models \contract'$. In addition, $(\system_3, r(10)) \models \lalways \neg \jboxO{\verdictE}{\euathlus} \pay$, and hence we obtain $(\system_3, r(10)) \models \court'$. Further, $(\system_3, r(10)) \models \ltime = 10 $.
	The
	result thus follows by the completeness theorem (Theorem \ref{thm:Weak Completeness-interpreted-neighborhood systems}). \qed
\end{proof}

It seems a safe assumption that Euathlus does not pay the fee at time 10 (the time of judge's pronouncement). Then, from this assumption and \eqref{eq: Protagoras and Euathlus arguments precise formalization}, using Lemma \ref{lem: basic results of linear temporal logic} and axiom $\funax$, we obtain:
\begin{align*}
	\Delta, \ltime = 10, \ltrue_{10} (\neg \pay) &\vdash_\CS 
	\lnext \jboxO{\verdictP}{\euathlus} \pay   \vee \lnext \jboxO{\agreement}{\euathlus} \pay.\\
	\Delta, \ltime = 10 &\vdash_\CS 
	\neg \lnext \jboxO{\agreement}{\euathlus} \pay \vee \neg  \lnext \jboxO{\verdictE}{\euathlus} \pay.
\end{align*}
Again, if the above arguments are formalized in $\Logic{SDL \oplus LTL}$, then we obtain
\begin{align}\label{eq: Protagoras and Euathlus arguments precise formalization 3}
	\Delta^\circ, \ltime = 10, \ltrue_{10} (\neg \pay) &\vdash_{\Logic{SDL \oplus LTL}} 
	\lnext \lobligatory \pay.\\
	\Delta^\circ, \ltime = 10 &\vdash_{\Logic{SDL \oplus LTL}} 
	\neg \lnext \lobligatory \pay.
\end{align}
This clearly yields to a contradiction:
\[
\Delta^\circ, \ltime = 10, \ltrue_{10} (\neg \pay) \vdash_{\Logic{SDL \oplus LTL}} \bot.
\]
In contrast to the standard doentic logic $\Logic{SDL \oplus LTL}$ in which a contradiction can be derived, we show that this set of assumption leads to no contradiction in $\JTO_\CS$.

\begin{lemma}\label{lem: second extended set of assumptions in the paradox is consistent}
	$\Delta \cup \lset{ \ltime = 10, \ltrue_{10} (\neg \pay) }$ is consistent in $\JTO_\CS$, where $\CS^\lobligatory = \emptyset$.
\end{lemma}
\begin{proof}
	Let $\CS$ be an arbitrary constant specification such that $\CS^\lobligatory = \emptyset$. Consider the interpreted-\Neigh system $\system_4 = (S, \runs, \neighborhood, \neighborhoodO,  \valuation)_{\agent \in \Ag}$ as follows: 
	\begin{itemize}
		\setlength\itemsep{0.1cm}
		\item 
		$S = \{ w_n \mid n \in \N\} \cup \lset{ v }$,
		
		\item 
		$\runs = \{ r  \}$, where $r : \N \to S$ is defined by $r(n) = w_n$ for all $n \in \N$,
		
		\item 
		$\neighborhood(r(n),t) = \{ S \}$, for all $t \in \Terms$, all $i \in \Ag$, and  all $n \in \N$, 
		%
		\item 
		$\neighborhoodO(r(n), \verdictP) = \lset{ \lset{ v } }$, for all  $i \in \Ag$, and  all $n \in \N$, \\ 
		$\neighborhoodO(r(n),x) = \neighborhoodO(r(n),c) = \emptyset$, for all $x \in \VTerms$ such that $x \neq \verdictP$, all $c \in \CTerms$, all $i \in \Ag$, and  all $n \in \N$, \\
		$  \neighborhoodO(r(n),t) = \lset{ \lset{ v } }$, for all $t \in \Terms \setminus \VTerms$, all $i \in \Ag$, and  all $n \in \N$,

		\item 
		$\valuation(r(n)) =  \lset{ \winP }$, for all $n \in \N$, and	$\valuation(v) = \lset{ \pay } \cup \lset{ \jboxOAgent{t} \phi \limplies \phi \mid \phi \in \Formulae }$.
	\end{itemize}
	%
	It is easy to show that $\system_4$ is an interpreted-\Neigh system for $\JTO_\CS$.
	For every $n \in \N$, we have that $(\system_4, r(n)) \models \lalways \lnext \neg \jboxO{\agreement}{\euathlus} \pay$. Thus, 
	$$
	(\system_4, r(n)) \models \lnext \neg \jboxO{\agreement}{\euathlus} \pay \lunless \winE,
	$$
	for every $n \in \N$. On the other hand, from $(\system_4, r(n)) \models \neg \winE$, for every $n \in \N$, it follows that
	$$
	(\system_4, r(n)) \models  \winE \limplies (\lnext  \jboxO{\agreement}{\euathlus} \pay \lunless \pay),
	$$
	for every $n \in \N$. Therefore, $(\system_4, r(10)) \models \contract'$. In addition, $(\system_4, r(10)) \models \lalways \lnext \jboxO{\verdictP}{\euathlus} \pay$, and hence $(\system_4, r(10)) \models \lnext \jboxO{\verdictP}{\euathlus} \pay \lunless \pay$. Thus, we obtain $(\system_4, r(10)) \models \court'$. Further, it is obvious that $(\system_4, r(10)) \models \ltime = 10 \wedge \ltrue_{10} (\neg \pay)$.
	The
	result thus follows by the completeness theorem (Theorem \ref{thm:Weak Completeness-interpreted-neighborhood systems}). \qed
\end{proof}

In addition to the above remarks, it is perhaps worth pointing out that contracts are usually considered to be common knowledge between agents. Thus, it is more appropriate to assume that $\contract'$ is common knowledge between Protagoras and Euathlus. Thus, instead of one single assumption $\contract'$, one may consider the following infinite set of assumptions  about the contract between Protagoras and  Euathlus:
\begin{eqnarray*}
	\contract'' := &\{ \contract' ,  \jbox{\agreement}_\protagoras \contract' , \jbox{\agreement}_\euathlus \contract', \jbox{\agreement}_\protagoras \jbox{\agreement}_\euathlus \contract', \jbox{\agreement}_\euathlus \jbox{\agreement}_\protagoras \contract', \\
	&\jbox{\agreement}_\protagoras \jbox{\agreement}_\euathlus \jbox{\agreement}_\protagoras \contract', \jbox{\agreement}_\euathlus \jbox{\agreement}_\protagoras \jbox{\agreement}_\euathlus \contract', \ldots \}.
\end{eqnarray*}
%
However, since we do not need the iteration of knowledge in our analysis, we continue with the single assumption $\contract'$ (or $\contract$).


\paragraph{\textbf{Does Protagoras permit to sue Euathlus?}}


Remember that according to the scenario Protagoras sues Euathlus at time 5, and thus he thought that he is permitted to sue Euathlus. The question that arises here is that: does Protagoras permit to sue Euathlus at time 5?  We first show that from the contract and the assumption that at time 5 Euathlus has not won his first court-case yet, it is concluded that there is no obligation for Euathlus to pay the fee until the time 5. According to the scenario at time 5 Euathlus has not won his first court-case yet, and this assumption can be formalized as follows:
\begin{equation*}\label{eq: Euathlus has not won his first court-case yet}
	\NowinE := \ltrue_{5} (\lsofar \neg \winE).
\end{equation*}
%

\begin{lemma}
	Let $\CS$ be an arbitrary constant specification for $\JTO$. Then
	\[
	 \contract, \NowinE \vdash_\CS \ltrue_5 (\lsofar \neg \jboxO{\agreement}{\euathlus} \pay).
	\]
\end{lemma}
\begin{proof}
	It is not difficult to show that
	$$
	\vdash_\CS \contract \wedge \neg \winE \limplies \neg \jboxO{\agreement}{\euathlus} \pay.
	$$
	By the rule $\sofarnecrule$ and $\sofarkax$, we obtain
	\begin{equation*}
			\vdash_\CS \lsofar \contract \wedge \lsofar \neg \winE \limplies \lsofar \neg \jboxO{\agreement}{\euathlus} \pay.
	\end{equation*}
	By Lemma \ref{lem: basic results of linear temporal logic} we have $\vdash_\CS  \contract \liff \lsofar \contract$, and hence 
	\begin{equation*}
		 \contract,  \lsofar \neg \winE \vdash_\CS   \lsofar \neg \jboxO{\agreement}{\euathlus} \pay.
	\end{equation*}
	Thus
	\begin{equation*}\label{eq: no obligation sofar}
		 \contract, \ltime = 5 \limplies \lsofar \neg \winE \vdash_\CS  \ltime = 5 \limplies \lsofar \neg \jboxO{\agreement}{\euathlus} \pay.
	\end{equation*}
	%
	%
	Then, using rules $\alwaysPastFuturenecrule$, we get 
	$$
	\lalwaysPastFuture \contract,  \ltrue_5 (\lsofar \neg \winE) \vdash_\CS \boxdot (\ltime = 5 \limplies \lsofar \neg \jboxO{\agreement}{\euathlus} \pay).
	$$
	By Lemma \ref{lem: basic results of linear temporal logic} we have $\vdash_\CS  \contract \liff \lalwaysPastFuture \contract$, and thus
	\[
	 \contract, \NowinE
	\vdash
	\ltrue_5 (\lsofar \neg \jboxO{\agreement}{\euathlus} \pay).
	\]
	\qed
\end{proof}

The above result shows that according to the contract there is no obligation for Euathlus to pay the fee until time 5 (i.e. the time that Protagoras sues Euathlus). On the other hand, we observe that Protagoras is permitted to sue Euathlus when (and only when) at some previous point Euathlus won his first court case and since that time he has not paid his fee. This assumption can be formalized as follows: 
%
\begin{equation*}\label{eq: Protagoras permit to sue Euathlus}
	\PsueE := \lalwaysPastFuture [\jboxP{\agreement}{\protagoras} \sueP \liff (\lsprevious  (\neg \pay \lsince \winE) \wedge \neg \pay)].
\end{equation*}
where $\sueP$ denotes the proposition ``Protagoras sues Euathlus." Next, it is easy to show that Protagoras is not even permitted to sue Euathlus.

\begin{lemma}
	Let  $\CS$ be an arbitrary constant specification for $\JTO$. Then
	\[
	\PsueE, \NowinE
	\vdash_\CS
	\ltrue_5 (\neg \jboxP{\agreement}{\protagoras} \sueP).
	\]
\end{lemma}
\begin{proof}
	We first show that
	\[
	\jboxP{\agreement}{\protagoras} \sueP \liff (\lsprevious  (\neg \pay \lsince \winE) \wedge \neg \pay), \lsofar \neg \winE 
	\vdash_\CS \neg \jboxP{\agreement}{\protagoras} \sueP.
	\]
	The proof is as follows:
	\begin{enumerate}
		\item $\jboxP{\agreement}{\protagoras} \sueP \liff (\lsprevious  (\neg \pay \lsince \winE) \wedge \neg \pay)$ \hfill hypothesis
		\item $\lsofar \neg \winE$ \hfill hypothesis
		\item $\lwprevious \lsofar  \neg \winE$ \hfill from 2 and Lemma \ref{lem: basic results of linear temporal logic}
		\item $\lwprevious \neg \lonce \neg \neg  \winE$ \hfill from 3 and the definition of $\lsofar$
		\item $\lwprevious \neg \lonce  \winE$ \hfill from 4 and Lemma \ref{lem: Admissible rules in LTL}
		\item $\neg \lsprevious  \lonce  \winE$ \hfill from 5 and the definition of $\lsprevious$
		\item $\neg \lsprevious (\neg \pay \lsince \winE)$ \hfill from 6 and axiom $\soneax$
		\item $\neg \lsprevious (\neg \pay \lsince \winE) \vee \pay$ \hfill from 7
		\item $\neg \jboxP{\agreement}{\protagoras} \sueP$ \hfill from 1 and 8
	\end{enumerate}
	%
	Thus
	$$
	\jboxP{\agreement}{\protagoras} \sueP \liff (\lsprevious  (\neg \pay \lsince \winE) \wedge \neg \pay), \ltime = 5 \limplies \lsofar \neg \winE
	 \vdash_\CS \ltime = 5 \limplies \neg \jboxP{\agreement}{\protagoras} \sueP.
	$$
	Then, using rules $\alwaysPastFuturenecrule$, we get 
	$$
	\PsueE, \ltrue_5 ( \lsofar \neg \winE) \vdash_\CS \lalwaysPastFuture (\ltime = 5 \limplies \neg \jboxP{\agreement}{\protagoras} \sueP).
	$$
	Therefore,
	\[
	\PsueE, \NowinE
	\vdash
	\ltrue_5 (\neg \jboxP{\agreement}{\protagoras} \sueP). 
	\]
	\qed
\end{proof}
This shows that from the assumption that at time 5 Euathlus has not won his first court case yet, it follows that Protagoras is not permitted to sue him. In fact, Protagoras brings suit against Euathlus on purpose and in order to fulfill the condition in the contract.

\subsection{Judges reasoning}
\label{sec:Judges reasoning}

In this section we show how judge can reason and what decision he should render. We suppose that court judgment is past-looking: the content of the ruling is based on the state of the world prior to the ruling.
\footnote{See \cite[page 72]{Goossens1977}  for a discussion on ``the	content of the ruling" and ``the ruling as an in-the-world event."} First note that the fact that judge's verdict at time 10 is against Protagoras can be expressed by the formula $\ltrue_{10} (\neg \winP)$. Moreover, since Euathlus has not won his first court case until time 10, instead of the assumption $\NowinE$, we can consider the following stronger assumption
$$
\NowinE' := \ltrue_{10} (\lsofar \lwprevious \neg \winE).
$$ 
Observe that  $\vdash_\CS \NowinE' \limplies \NowinE$.

As it is mentioned in \cite{Jankowski2015}, a court decision can only take into
account what has happened up to the pronouncement of a judgment at the latest. Following Goossens \cite{Goossens1977} let me shorten this legal fact as follows:

\[
\text{Court judgement is past-looking.} \qquad \qquad (*)
\]

In the logic $\JTO$, $(*)$ can be expressed as the following assumption:
\[
\PastLooking := \ltrue_{10} (\lsofar \lwprevious \neg \winE \limplies \neg \winP).
\]
Up to time 10, Euathlus has not won a case yet, only with the pronouncement
itself (at time 10) the condition could possibly be met (if judgment is given against Protagoras). Therefore, from a legal point of view,
the judge should have no problem deciding in favor of Euathlus. In fact, the following set of assumptions 
$$
 \{\contract, \court, \NowinE', \PastLooking \}
$$
is consistent. This can be shown by means of the interpreted-neighborhood system $\system_1$ which is defined in the proof of Lemma \ref{lem:the set of assumptions in the paradox is consistent}. It is not difficult to show that 
$$
(\system_1, r(0)) \models \contract \wedge \court \wedge \NowinE' \wedge \PastLooking.
$$
This is  the standard solution given by Leibniz (\cite{Jankowski2015}), 
 Lenzen \cite{Lenzen1977}, Smullyan \cite{Smullyan1978}, and others.  

\begin{theorem}[Judge's verdict in the first case]\label{thm: Judges verdict in the first case}
	Let  $\CS$ be an arbitrary constant specification for $\JTO$. Then
	\[
	\contract, \court, \NowinE', \PastLooking \vdash_\CS  \ltrue_{10} (\neg \winP \wedge \winE \wedge \jboxO{\agreement}{\euathlus} \pay).
	\]	
\end{theorem}
\begin{proof}
	The proof is as follows:
	\begin{enumerate}
		\item $\ltime = 10$ \hfill hypothesis
		\item $\lsofar \lwprevious \neg \winE$ \hfill from 1 and $\NowinE'$
		\item $\neg \winP$ \hfill from  2 and $\PastLooking$
		\item $\winE$ \hfill from 1, 3, and $\court$
		\item $\jboxO{\agreement}{\euathlus} \pay$ \hfill from 4 and $\contract$
		\item $\neg \winP \wedge \winE \wedge \jboxO{\agreement}{\euathlus} \pay$ \hfill from 3, 4, and 5 
	\end{enumerate}
	Thus, we proved that
	$$
	\contract, \court, \NowinE', \PastLooking 
	 \vdash_\CS \ltime = 10 \limplies \neg \winP \wedge \winE \wedge \jboxO{\agreement}{\euathlus} \pay.
	$$
	Then, using rules $\alwaysPastFuturenecrule$, we get 
	$$
	\lalwaysPastFuture \contract, \lalwaysPastFuture \court, \lalwaysPastFuture \NowinE', \lalwaysPastFuture \PastLooking 
	 \vdash_\CS \lalwaysPastFuture (\ltime = 10 \limplies \neg \winP \wedge \winE \wedge \jboxO{\agreement}{\euathlus} \pay).
	$$
	By Lemmas \ref{lem: basic results of linear temporal logic} and \ref{lem: temporal truth predicate properties}, we obtain
	\[
	\contract, \court, \NowinE', \PastLooking
	 \vdash_\CS  \ltrue_{10} (\neg \winP \wedge \winE \wedge \jboxO{\agreement}{\euathlus} \pay). 
	\]
	\qed
\end{proof}


Now let us continue the scenario in the following way. At time 10, judge's verdict is in favor of Protagoras, and Euathlus has not paid the fee until, say, time 15. Then, at time 15, Protagoras sues Euathlus again (this is the second case), and we show that this time he is permitted to sue Euathlus.

\begin{lemma}\label{lem: Protagoras is permitted to sue in the second case}
	Let  $\CS$ be an arbitrary constant specification for $\JTO$. Then
	\[
	\ltrue_{10} (\winE ), \ltrue_{15} (\lsofar \neg \pay), \PsueE
	\vdash_\CS
	\ltrue_{15} (\jboxP{\agreement}{\protagoras} \sueP).
	\]
\end{lemma}
\begin{proof}
	The proof is as follows:
	\begin{enumerate}
		\item $\ltrue_{10} (\winE )$ \hfill hypothesis
		\item $\ltrue_{15} (\lsofar \neg \pay)$ \hfill hypothesis
		\item $\PsueE$ \hfill hypothesis
		\item $\ltime = 15$ \hfill hypothesis
		\item $\ltrue_{11} (\lwprevious \winE )$ \hfill from  1 and Lemma \ref{lem: temporal truth predicate properties}
		\item $\ltime = 11 \limplies \lwprevious \winE$ \hfill from 5
		\item $\ltime = 11 \limplies (\neg \pay \lsince \lwprevious \winE)$ \hfill from 6 and $\stwoax$
		\item $\ltime = 12 \limplies \lsprevious (\neg \pay \lsince \lwprevious \winE)$ \hfill from 7 and Lemma \ref{lem: temporal truth predicate properties}
		\item $\ltrue_{12} (\lsofar \neg \pay)$ \hfill from 2 and Lemma \ref{lem: temporal truth predicate properties}
		\item $\ltime = 12 \limplies \neg \pay$ \hfill from 9
		\item $\ltime = 12 \limplies \neg \pay \wedge \lsprevious (\neg \pay \lsince \lwprevious \winE)$ \hfill from 8 and 10 
		\item $\ltime = 12 \limplies  \neg \pay \lsince \lwprevious \winE$ \hfill from 11 and $\stwoax$
		\item $\ltime = 13 \limplies  \lsprevious (\neg \pay \lsince \lwprevious \winE)$ \hfill from 12 and Lemma \ref{lem: temporal truth predicate properties}
		\item $\ltime = 13 \limplies \neg \pay$ \hfill from 2 and Lemma \ref{lem: temporal truth predicate properties} (similar to step 10)
		\item $\ltime = 13 \limplies \neg \pay \wedge \lsprevious (\neg \pay \lsince \lwprevious \winE)$ \hfill from 13 and 14 
		\item $\ltime = 13 \limplies  \neg \pay \lsince \lwprevious \winE$ \hfill from 15 and $\stwoax$
		\item $\ltime = 14 \limplies  \lsprevious (\neg \pay \lsince \lwprevious \winE)$ \hfill from 16 and Lemma \ref{lem: temporal truth predicate properties}
		\item $\ltime = 14 \limplies \neg \pay$ \hfill from 2 and Lemma \ref{lem: temporal truth predicate properties} (similar to step 10)
		\item $\ltime = 14 \limplies \neg \pay \wedge \lsprevious (\neg \pay \lsince \lwprevious \winE)$ \hfill from 17 and 18
		\item $\ltime = 14 \limplies  (\neg \pay \lsince \lwprevious \winE)$ \hfill from 19 and $\stwoax$
		\item  $\ltime = 15 \limplies \lsprevious (\neg \pay \lsince \lwprevious \winE)$ \hfill from 20 and Lemma \ref{lem: temporal truth predicate properties}
		\item $\ltime = 15 \limplies \neg \pay$ \hfill from 2 and Lemma \ref{lem: temporal truth predicate properties} (similar to step 10)
		\item $\ltime = 15 \limplies \neg \pay \wedge \lsprevious (\neg \pay \lsince \lwprevious \winE)$ \hfill from 21 and 22
		\item  $\ltime = 15 \limplies  \jboxP{\agreement}{\protagoras} \sueP$ \hfill from 3 and 23
	\end{enumerate}
	Thus, we prove that
	$$
	\ltrue_{10} (\winE ), \ltrue_{15} (\lsofar \neg \pay), \PsueE
	 \vdash_\CS \ltime = 15 \limplies \jboxP{\agreement}{\protagoras} \sueP.
	$$
	Then, using rule $\alwaysPastFuturenecrule$, we get 
	$$
	\ltrue_{10} (\winE ), \ltrue_{15} (\lsofar \neg \pay), \PsueE
	\vdash_\CS
	\ltrue_{15} (\jboxP{\agreement}{\protagoras} \sueP).
	$$
	\qed
\end{proof}

Let the atomic proposition $\winPsecond$ denote the sentence ``Protagoras wins the second case (or judges award the second case to Protagoras)." Here, we assume that at the time of the second case's pronouncement if Euathlus has won already his first case and since then he has not paid yet, then judge can award the second case to Protagoras. This assumption can be formalized as follows;
\[
\ltrue_{15} [\lsprevious (\neg \pay \lsince \lwprevious \winE) \wedge \neg \pay \limplies \winPsecond].
\]
Next we show that at time 15 (the time of the second case pronouncement) judge awards the second case to Protagoras. 

\begin{theorem}[Judge's verdict in the second case]\label{thm: Judges verdict in the second case}
	Let  $\CS$ be an arbitrary constant specification for $\JTO$. Then
	 \begin{eqnarray*}
	 	\ltrue_{10} (\winE ),& \ltrue_{15} (\lsofar \neg \pay), \PsueE, \ltrue_{15} [\lsprevious (\neg \pay \lsince \lwprevious \winE) \wedge \neg \pay \limplies \winPsecond] \\
	 	 &\vdash_\CS \ltrue_{15} (\winPsecond)
	 \end{eqnarray*}
\end{theorem}
\begin{proof}
	Follows from Lemma \ref{lem: Protagoras is permitted to sue in the second case}. \qed
\end{proof}

Therefore, one solution to this form of the paradox is that, provided that (*) holds, if Euathlus defends himself, he should have won the first case, since he hasn't yet won his first case (as stated in Theorem \ref{thm: Judges verdict in the first case}). Protagoras could have sued a second time and won (as stated in Theorem \ref{thm: Judges verdict in the second case}). \footnote{Euathlus might then sue Protagoras for malicious prosecution and ask compensation. }
By bringing the first suit, even if Protagoras is certain to lose it,
he guarantees his victory in the second suit.

As mentioned before, our analysis of the paradox is compatible with Leibniz' viewpoint. Leibniz writes:

\begin{quote}
	... under the first arrangement Protagoras loses, under the last he wins. For
	when Protagoras demands payment from his student before the day on
	which it is owed and can be demanded and before the condition has been
	fulfilled (the condition of his payment is this: victory in one’s first case),
	without doubt he will be seen to have made his demand too soon. Therefore,
	in the first instance, the case is lost in the immediate circumstance,
	and by that very fact the condition of payment will be fulfilled, because
	Euathlus has won his first case. (quoted from \cite[page 8]{Sobel1987}.)
\end{quote}

Therefor, we accept the solution given by Leibniz (\cite{Jankowski2015}), Lenzen \cite{Lenzen1977}, Aqvist \cite{Aqvist1995}, Smullyan \cite{Smullyan1978}, and others, that Euathlus wins the first case (on the ground that Protagoras' demand for the payment was premature), and that Protagoras can win the second case (provided Euathlus does not pay the fee after the first case and Protagoras sue him again). So the Protagoras versus Euathlus case is not a perplex one, and the appearance of a contradiction in the arguments of Protagoras and Euathlus is an illusion promoted by a mis-usage in the grounds of obligations.

\bibliography{library}

\end{document}